\documentclass[11pt,reqno]{article}
\usepackage{amssymb,amscd,amsmath,amsthm}
\newtheorem{thm}{Theorem}[section]
\newtheorem{prop}[thm]{Proposition}
\newtheorem{lemma}[thm]{Lemma}
\newtheorem{cor}[thm]{Corollary}

\newtheorem{remark}[thm]{Remark}
\newtheorem{example}[thm]{Example}

\def\bN{\mathbb{N}}
\def\bM{\mathbb{M}}
\def\bH{\mathbb{H}}
\def\bP{\mathbb{P}}
\def\Tr{\mathrm{Tr}\,}
\def\HS{\mathrm{HS}}
\def\im{\mathrm{i}}
\def\<{\langle}
\def\>{\rangle}
\def\pont{\,\cdot\,}
\def\bL{\mathbb{L}}
\def\bbbr{\mathbb{R}}
\def\bR{\mathbb{R}}
\def\bJ{\mathbb{J}}
\def\cD{\mathcal{D}}
\def\diag{\mathrm{Diag}}
\def\eps{\varepsilon}
\def\MM{\frak{M}}
\def\bI{\mathbb{I}}

\begin{document}

\ \vskip 1cm 
\centerline{\LARGE {\bf Riemannian metrics}}
\bigskip
\centerline{\LARGE {\bf on positive definite matrices}}
\bigskip
\centerline{\LARGE {\bf related to means}} 
\bigskip
\bigskip
\centerline{\Large
Fumio Hiai\footnote{E-mail: hiai@math.is.tohoku.ac.jp}
and D\'enes Petz\footnote{E-mail: petz@math.bme.hu}}

\medskip
\begin{center}
$^1$\,Graduate School of Information Sciences, Tohoku University \\
Aoba-ku, Sendai 980-8579, Japan
\end{center}
\begin{center}
$^2$\,Alfr\'ed R\'enyi Institute of Mathematics, \\ H-1364 Budapest,
POB 127, Hungary
\end{center}

\medskip
\begin{abstract}
The Riemannian metric on the manifold of positive definite matrices is defined by a
kernel function $\phi$ in the form
$K_D^\phi(H,K)=\sum_{i,j}\phi(\lambda_i,\lambda_j)^{-1}\Tr P_i H P_j K$ when
$\sum_i \lambda_i P_i$ is the spectral decomposition of the foot point $D$ and the
Hermitian matrices $H,K$ are tangent vectors. For such kernel metrics the tangent space
has an orthogonal decomposition. The pull-back of a kernel metric under a mapping
$D \mapsto G(D)$ is a kernel metric as well. Several Riemannian geometries of the
literature are particular cases, for example, the Fisher-Rao metric for multivariate
Gaussian distributions and the quantum Fisher information. In the paper the case
$\phi(x,y)=M(x,y)^\theta$ is mostly studied when $M(x,y)$ is a mean of the positive
numbers $x$ and $y$. There are results about the geodesic curves and geodesic distances.
The geometric mean, the logarithmic mean and the root mean are important cases.
 
\bigskip\noindent
{\it AMS classification:}
15A45; 15A48; 53B21; 53C22

\medskip\noindent
{\it Keywords:}
positive definite matrix; Riemannian metric; Fisher-Rao metric; quantum Fisher
information; quantum skew information; symmetric homogeneous mean; logarithmic mean;
geometric mean; geodesic curve; geodesic distance; Fr\'echet derivative; divided
difference
\end{abstract}

\medskip
\section*{Introduction}
The $n \times n$ positive definite matrices with complex entries can be parametrized
by the real and imaginary parts of the entries, and they form an open subset of the space
$\bH_n$ of $n\times n$ Hermitian matrices regarded as the 
Euclidean space $\bbbr^m$, where $m=n^2$. Hence the tangent space of their manifold
$\bP_n$ at any foot point can be identified with $\bH_n$. A Riemannian metric $K_D(H,K)$
is a family of inner products on  $\bH_n$ depending smoothly on the foot point $D$.
If $\phi(x,y)$ is a positive kernel function on $(0,\infty)\times(0,\infty)$ and $D$ has
the spectral decomposition $\sum_{i=1}^k\lambda_iP_i$, then a Riemannian metric can be
defined as
\begin{equation}\label{F-0.1}
K_D^\phi(H,K):=\sum_{i,j=1}^k\phi(\lambda_i,\lambda_j)^{-1}\Tr P_i H P_j K,
\end{equation}
where $\Tr$ is the usual trace on matrices. The goal of the present paper is to study
this kind of Riemannian metrics.

As far as the authors know, the first example of \eqref{F-0.1} is historically the case
$\phi(x,y)=xy$ which was considered by Skovgaard \cite{Sk} as a Fisher-Rao statistical
Riemannian metric on positive definite matrices describing multivariate Gaussian
distributions. Another example is also related to Fisher information. In the quantum
mechanical setting the states correspond to positive semidefinite  matrices of trace 1,
and in \cite{Pe3,PS} the metric \eqref{F-0.1} was justified in the particular case
$\phi(x,y)=yf(x/y)$, where $f:(0,\infty)\to  (0,\infty)$ is an operator monotone function.
More details on these examples are presented in the rest of this section.

The trivial choice $\phi(x,y)\equiv 1$ gives a flat space where the Riemannian metric
is the Hilbert-Schmidt inner product $\<H,K\>_\HS$ on $\bH_n$. The Hilbert-Schmidt inner
product $\<X,Y\>_\HS:=\Tr X^*Y$ and the Hilbert-Schmidt norm $\|X\|_\HS:=(\Tr X^*X)^{1/2}$
are defined on the space $\bM_n$ of all $n\times n$ complex matrices, and the space
$(\bH_n,\<\cdot,\cdot\>_\HS)$ is a real subspace of the Hilbert space
$(\bM_n,\<\cdot,\cdot\>_\HS)$.

The positive definite real matrices might be considered as the variance of multivariate
normal distributions and the information geometry of Gaussians yields a natural Riemannian
metric. The simplest way to construct an information geometry is to start with an
information potential function and to introduce the Riemannian metric by the Hessian of
the potential. We want a geometry on the family of non-degenerate multivariate Gaussian
distributions with zero mean vector. Those distributions are given by a positive definite
real matrix $D$ in the form
$$
p_D(x):=\frac{1}{\sqrt{(2\pi)^{n}\det D}} \exp\biggl(-{\<D^{-1}x,x\>\over2}\biggr),
\qquad x \in \bR^n.
$$
We identify the Gaussian $p_D$ with the matrix $D$, and we can say that the
Riemannian geometry is constructed on the space of positive definite real matrices. There
are many reasons (originated from statistical mechanics, information theory and
mathematical statistics) that the Boltzmann entropy
$$
S(p_D):=\frac{1}{2}\log(\det D)+\mbox{const.}
$$
is a candidate for being an information potential.

The $n\times n$ real symmetric matrices can be identified with the Euclidean space of
dimension $n(n+1)/2$ and the positive definite matrices form an open subset. Therefore the
set of Gaussians has a simple and natural manifold structure. The tangent space at each
foot point is the set of symmetric matrices. The Riemannian metric is defined as the
Hessian
$$
g_D(H,K):=\frac{\partial^2}{\partial s\partial t}
S(p_{D+sH+tK})\Big|_{s=t=0}\,,
$$
where $H$ and $K$ are tangents at $D$. The differentiation easily gives
\begin{equation}\label{F-0.2}
g_D(H,K)=\Tr D^{-1}HD^{-1}K.
\end{equation}
The corresponding information geometry of the Gaussians was discussed in \cite{OSA} in
detail. In the statistical model of multivariate Gaussian distributions, \eqref{F-0.2}
plays the role of the {\it Fisher-Rao metric}. We note here that this geometry has many
symmetries. Each congruence transformation of the matrices becomes a symmetry, namely
\begin{equation}\label{F-0.3}
g_{TDT^t}(THT^t,TKT^t)=g_D(H,K)
\end{equation}
for every real invertible matrix $T$.

Formula \eqref{F-0.2} determines a Riemannian metric on the manifold $\bP_n$ as well and
below we prefer to consider the complex case. Note that if we want to find the geodesic
curve between $A$ and $B$, then it is sufficient to find the geodesic joining $I$ and
$A^{-1/2}BA^{-1/2}$ due to property \eqref{F-0.3}. This is essentially easier since
they commute. In fact, concerning the geodesic curves in the Riemannian manifold
$(\bP_n,g)$, it is known \cite{LL,Mo,BH} that for each $A,B\in\bP_n$ there exists a unique
geodesic shortest curve joining $A,B\in\bP_n$ given by
\begin{equation}\label{F-0.4}
\gamma(t)=A\,\#_t\,B:=A^{1/2}(A^{-1/2}BA^{-1/2})^t A^{1/2},\qquad 0\le t\le 1,
\end{equation}
and the geodesic midpoint $\gamma(1/2)$ is just the {\it geometric mean} (\cite {PW,An})
$$
A\,\#\,B:=A^{1/2}(A^{-1/2}BA^{-1/2})^{1/2}A^{1/2}.
$$
Furthermore, the geodesic distance is
\begin{equation}\label{F-0.5}
\delta(A,B)=\|\log(A^{-1/2}BA^{-1/2})\|_\HS.
\end{equation}
In this way, the information Riemannian geometry is adequate to treat the geometric mean
of positive definite matrices.

For each $A,B\in\bP_n$ the mean $C':=A\,\#\,B$ is the midpoint of the geodesic joining
$A$ and $B$, $A':=B\,\# C\,$ and $B':=C\,\#\,A$ are similar. Since
$\delta(B\,\#\,C,C\,\#\,A)\le{1\over2}\delta(A,B)$ by \cite[Proposition 6]{BH}, the diameter
of the triangle $A'B'C'$ is at most the half of the diameter of $ABC$. This result gives
a geometric proof of the recursive construction
of geometric mean of 3 positive matrices proposed in \cite{ALM}. Note that another
``geometric mean" of $A_1,\dots,A_k\in\bP_n$ was introduced in \cite{Mo,BH} as the unique
minimizer of $A\in\bP_n\mapsto\sum_{j=1}^k\delta^2(A,A_j)$.

We denote by $\cD_n$ the set of all $n\times n$ positive definite matrices of trace 1,
which is a smooth differentiable manifold as a submanifold of $\bP_n$. The tangent space
of the manifold $\cD_n$ at each foot point $D$ is the subspace of $\bH_n$ consisting of
$n\times n$ Hermitian matrices of trace $0$, i.e.,
$T_D\cD_n=\bH_n\ominus\bR I:=\{H\in\bH_n:\Tr H=0\}$.
One can define a Riemannian metric on $\cD_n$ in the form
$$
K_D(H,K)=\<H,\bJ_D^{-1}K\>_\HS,\qquad D\in\cD_n,\ H,K\in\bH_n\ominus\bR I,
$$
where $\bJ_D$ is a positive linear operator on the real Hilbert space
$(\bH_n\ominus\bR I,\<\cdot,\cdot\>_\HS)$. One can extend $\bJ_D$ to a positive symmetric
operator on $\bH_n$ and furthermore to a positive operator on the Hilbert space
$(\bM_n,\<\cdot,\cdot\>_\HS)$ by complexification. So we may assume that a Riemannian
metric $K_D$ is given on $\cD_n$, $n\in\bN$, by
$K_D(X,Y)=\<X,\bJ_D^{-1}Y\>_\HS$ for $X,Y\in\bM_n$.
The metric $K_D$ (more precisely, a sequence of metrics $K_D$ on $\cD_n$, $n\in\bN$) is
{\it monotone} if for any completely positive and trace preserving map (or coarse graining)
$\beta:\bM_n\to\bM_m$ we have
$$
K_{\beta(D)}(\beta(X),\beta(X))\le K_D(X,X),\qquad D\in\cD_n,\ X\in\bN_n.
$$
Recall that $\beta$ is completely positive and
trace preserving if and only if $\beta^*$ is completely positive and unital. It was proved
in Petz \cite{Pe3} that the monotone metrics $K_D$ with normalization
$K_D(I,I)=\Tr(D^{-1})$ correspond one-to-one to the operator monotone functions
$f:(0,\infty)\to(0,\infty)$ with normalization $f(1)=1$ as follows:
\begin{equation}\label{F-0.6}
K_D^f(X,Y):=\<X,(\bJ_D^f)^{-1}Y\>_\HS\quad\mbox{and}\quad
\bJ_D^f:=f(\bL_D\bR_D^{-1})\bR_D.
\end{equation}
Furthermore, $K_D^f$ is symmetric if and only if $f$ is symmetric, i.e., $xf(x^{-1})=f(x)$,
$x>0$. We say that an operator monotone function $f\ge0$ on $(0,\infty)$ is {\it standard}
if $f(1)=1$ and $xf(x^{-1})=f(x)$.

On the other hand, the theory of {\it operator means} due to Kubo and Ando \cite{KA} says
that there is a one-to-one correspondence between the symmetric operator means (or matrix
means) and the standard operator monotone functions $f$ as follows:
$$
\sigma_f(A,B):=A^{1/2}f(A^{-1/2}BA^{-1/2})A^{1/2},\qquad A,B\in\bP_n.
$$
Thus one may write
\begin{equation}\label{F-0.7}
K_D^f(X,Y)=\<X,\sigma_f(\bL_D,\bR_D)^{-1}Y\>_\HS.
\end{equation}
When $D=\diag(\lambda_1,\dots,\lambda_n)$ is diagonal, one can more explicitly write
$$
K_D^f(X,X)=\sum_{i,j=1}^n{1\over\lambda_jf(\lambda_i/\lambda_j)}|X_{ij}|^2,
\qquad X=[X_{ij}]\in\bM_n.
$$
For each standard operator monotone function $f$, the symmetric monotone metric (or the
{\it quantum Fisher information}) $K_D^f$ originally defined on $\cD_n$ by \eqref{F-0.6}
or \eqref{F-0.7} can be automatically extended to $\bP_n$ by the same formula.

It was also observed in Lesniewski and Ruskai \cite{LR} that any of the above metrics $K^f$
can be realized as the Hessian
$$
K_D^f(H,K)=-{\partial^2\over\partial s\partial t}
S_F(D+sH,D+tK)\Big|_{s=t=0},
$$
of a quasi-entropy \cite{Pe1,Pe2} $S_F(D_1,D_2)$ defined by a function $F$ on $(0,\infty)$
with the relation $1/f(x)=(F(x)+xF(x^{-1}))/(x-1)^2$.

The {\it Wigner-Yanase-Dyson skew information} is the quantity
$$
I_D^\mathrm{WYD}(p,K):=-{1\over2}\Tr[D^p,K][D^{1-p},K],
\qquad D\in\cD_n,\ K\in\bH_n,
$$
where $0<p<1$. The case $p=1/2$ is the original Wigner-Yanase skew information. It was
observed in \cite{PH} that the Wigner-Yanase-Dyson skew information $I_D^\mathrm{WYD}(p,K)$
coincides, apart from a constant factor, with a monotone Riemannian metric
$$
K_D^{f_p}(\im[D,K],\im[D,K]),
$$
where $f_p$ is a standard operator monotone function defined by
\begin{equation}\label{F-0.8}
f_p(x):=p(1-p){(x-1)^2\over(x^p-1)(x^{1-p}-1)}.
\end{equation}
The notion of skew information was recently generalized by Hansen \cite{Ha1} as follows:
For each standard operator monotone function $f$ that is regular, i.e.,
$f(0)\ (:=\lim_{x\searrow0}f(x))>0$, the {\it the metric adjusted skew information} (or
the {\it quantum skew information}) corresponding to
$f$ is
\begin{equation}\label{F-0.9}
I_D^f(K):={f(0)\over2}K_D^f(\im[D,K],\im[D,K]),
\qquad D\in\cD_n,\ K\in\bH_n,
\end{equation}
which is explicitly written as
$$
I_D^f(K)={f(0)\over2}\sum_{i,j=1}^n
{(\lambda_i-\lambda_j)^2\over\lambda_jf(\lambda_i/\lambda_j)}|K_{ij}|^2
$$
if $D=\diag(\lambda_1,\dots,\lambda_n)$.

Via the operator $\bJ_D^f$ in \eqref{F-0.6}, each standard operator monotone function $f$
defines a quantity
\begin{equation}\label{F-0.10}
\varphi_D[K,K]:=\<K,\bJ_D^fK\>_\HS,\qquad D\in\cD_n,\ K\in\bH_n,
\end{equation}
which was called {\it generalized variance} in \cite{Pe4}. Any such variance has the
property $\varphi_D[K,K]=\Tr DK^2$ for commuting $D$ and $K$.

In the present paper we study Riemannian geometry on $\bP_n$ with kernel metrics $K^\phi$
in \eqref{F-0.1} when the kernel function $\phi(x,y)$ is in the form $M(x,y)^\theta$, a
degree $\theta\in\bR$ power of a certain mean $M(x,y)$ for two positive numbers (as
prescribed at the beginning of Section 2). The above quantities \eqref{F-0.2},
\eqref{F-0.6} and \eqref{F-0.10} are important special cases where $\theta=2,1$ and $-1$,
respectively. The paper is organized as follows. After describing our setting in Section 1
in more detail, in Section 2 we determine Riemannian metrics in our class which are
written as a pull-back of the Euclidean metric. For such metrics the geodesic curve and
the geodesic distance are explicitly given (Theorem \ref{T-2.1}). Section 3 is concerned
with the (non-)completeness of Riemannian metrics in our class (Theorem \ref{T-3.1}) and
pull-back metrics from the Fisher-Rao metric $g$ (Theorem \ref{T-3.3}). In Section 4 we
discuss comparison properties among our Riemannian metrics. The comparison of geodesic
distances for two metrics is easily described in terms of the corresponding means and the
degrees of power (Theorem \ref{T-4.1}). Finally in Section 5, we treat the generalized
situation (of Finsler metrics rather than Riemannian metrics) where unitarily invariant
norms are applied in place of the Hilbert-Schmidt norm.

For basics on Riemannian geometry the reader may refer to texts \cite{KN, Michor} for
example.

\section{Riemannian metrics induced by kernel functions}
\setcounter{equation}{0}

For each $D\in\bP_n$ the {\it left} and {\it right multiplication} operators $\bL_D$ and
$\bR_D$   are defined as $\bL_DX:=DX$ and $\bR_DX:=XD$ for $X\in\bM_n$. Note that $\bL_D$
and $\bR_D$ are commuting positive operators on the Hilbert space
$(\bM_n,\<\pont,\pont\>_\HS)$, i.e., $\bL_D\bR_D=\bR_D\bL_D$, $\<X,\bL_DX\>_\HS\ge0$ and
$\<X,\bR_DX\>_\HS\ge0$ for all $X\in\bM_n$.  For a kernel function
$\phi:(0,\infty)\times(0,\infty)\to(0,\infty)$, a positive  operator $\phi(\bL_D,\bR_D)$
on $(\bM_n,\<\pont,\pont\>_\HS)$ is defined via  functional calculus, that is, when
$D=\sum_{i=1}^k\lambda_iP_i$ is the spectral decomposition,
$$
\phi(\bL_D,\bR_D)X:=\sum_{i=1}^k\phi(\lambda_i,\lambda_j)P_iXP_j,\qquad X\in\bM_n.
$$
When $\phi(x,y)$ is smooth in $x$ and $y$, one can define a {\it Riemannian metric}
$K^\phi$ on $\bP_n$ by
\begin{equation}\label{F-1.1}
K_D^\phi(H,K):=\<H,\phi(\bL_D,\bR_D)^{-1}K\>_\HS
=\sum_{i,j=1}^k\phi(\lambda_i,\lambda_j)^{-1}\Tr P_i H P_j K
\end{equation}
when $H,K\in\bH_n$. 

By taking the diagonalization
$D=U\diag(\lambda_1,\dots,\lambda_n)U^*$ with a unitary $U$, one can also write
\begin{equation}\label{F-1.2}
\phi(\bL_D,\bR_D)^{-1/2}H
=U\Biggl(\Biggl[{1\over\sqrt{\phi(\lambda_i,\lambda_j)}}\Biggr]_{ij}
\circ(U^*HU)\Biggr)U^*,
\end{equation}
where $\circ$ denotes the {\it Schur} (or {\it Hadamard}\,) {\it product}\,.

\begin{lemma}\label{L-1.1}
For each $D\in\bP_n$ let
$$
T_D^c:=\{H\in\bH_n:HD=DH\}\quad \mbox{and}\quad
T_D^q:=\{\im[D,K]:K\in\bH_n\}.
$$
Then
\begin{itemize}
\item[\rm(1)] $K_D^\phi(H,K)=\Tr\hat\phi(D)HK$ if $H \in T_D^c$ and $K\in\bH_n$, where
$\hat\phi(x):=1/\phi(x,x)$, $x>0$.
\item[\rm(2)] $K_D^\phi(H,\im[D,K])=0$ if $H \in T_D^c$ and $K\in\bH_n$.
\item[\rm(3)] $K_D^\phi(\im[D,K],\im[D,K])=\<K,\tilde\phi(\bL_D,\bR_D)K\>_\HS$ for all
$K\in\bH_n$, where 
$$
\tilde\phi(x,y):=\frac{(x-y)^2}{\phi(x,y)},\qquad x,y>0.
$$
\end{itemize}
In particular, the tangent space $T_D=\bH_n$ has an orthogonal decomposition 
$T_D=T_D^c\oplus T_D^q$ with respect to $K_D^\phi$.
\end{lemma}

The proof of the lemma is left to the reader, which is easy by using \eqref{F-1.1}.

When $\gamma:[0,1]\to\bP_n$ is a $C^1$ curve (or more generally, a continuous and
piecewise $C^1$ curve), the {\it length} of $\gamma$ with respect to the metric $K^\phi$
is given by
\begin{equation}\label{F-1.3}
L_\phi(\gamma):=\int_0^1\sqrt{K_{\gamma(t)}^\phi(\gamma'(t),\gamma'(t))}\,dt
=\int_0^1\|\phi(\bL_{\gamma(t)},\bR_{\gamma(t)})^{-1/2}\gamma'(t)\|_\HS\,dt.
\end{equation}
Note that the length $L_\phi(\gamma)$ is independent of the choice of the parametrization
of $\gamma$. The {\it geodesic distance} $\delta_\phi(A,B)$ between $A,B\in\bP_n$ is the
infimum of $L_\phi(\gamma)$ over all $C^1$ curves (or equivalently, over all smooth curves)
$\gamma$ from $A$ to $B$. A {\it geodesic shortest curve} is a curve from $A$ to $B$ such
that $L_\phi(\gamma)=\delta_\phi(A,B)$.

Now let $G$ be a smooth function from an open interval $(a,b)$ into $(0,\infty)$. Assume
that $G'(x)\ne0$ for all $x\in(a,b)$ so that $G$ is a diffeomorphism from $(a,b)$ onto a
subinterval of $(0,\infty)$. Let $\bH_n(a,b)$ denote the submanifold $\{A\in\bH_n:a<A<b\}$
of $\bH_n$, where $a<A<b$ means that all the eigenvalues of $A$ are in $(a,b)$. Then the
map $A\mapsto G(A)$ defined via functional calculus is a smooth diffeomorphism from
$\bH_n(a,b)$ into $\bP_n$. Our next aim is to determine a Riemannian metric on $\bH_n(a,b)$
such that $A\mapsto G(A)$ is an isometry into the Riemannian space $(\bP_n,K^\phi)$. This
Riemannian metric on $\bH_n(a,b)$ is called the {\it pull-back} of $K^\phi$ under the
transformation $A\mapsto G(A)$.

\begin{lemma}\label{L-1.2}
Let $K_A$, $A\in\bH_n(a,b)$, be the pull-back of the Riemannian metric $K^\phi$ on $\bP_n$
under $A\mapsto G(A)$ as mentioned above. Let $A\in\bH_n(a,b)$ and
$A=\sum_{i=1}^k\lambda_iP_i$ be the spectral decomposition. Furthermore, let
$T_A^c:=\{H\in\bH_n:HA=AH\}$ as in Lemma \ref{L-1.1}. Then
\begin{itemize}
\item[\rm(1)] ${d\over dt}G(A+tH)\big|_{t=0}=G'(A)H$ if $H\in T_A^c$.
\item[\rm(2)] ${d\over dt}G(A+t\im[A,K])\big|_{t=0}=\im[G(A),K]$ for all $K\in\bH_n$. 
\item[\rm(3)] For every $H\in T_A^c$,
$$
K_A(H,H)=\sum_{i=1}^k\frac{G'(\lambda_i)^2}
{\phi(G(\lambda_i),G(\lambda_i))}\Tr P_iH^2.
$$
\item[\rm(4)] For every $H\in T_A^c$ and $K\in\bH_n$, $K_A(H,\im[D,K])=0$.
\item[\rm(5)] For every $K\in\bH_n$,
$$
K_A(\im[A,K],\im[A,K])=\sum_{i=1}^k\frac{(G(\lambda_i)-G(\lambda_j))^2}
{\phi(G(\lambda_i),G(\lambda_j))}\Tr P_i K P_jK.
$$
\end{itemize}
\end{lemma}

\begin{proof}
(1) is obvious.

(2)\enspace
This is found in \cite{Pe-book} but a short proof using the differential formula (see
\cite{Bh}) is given here. We may assume without loss of generality that $A$ is diagonal as
$A=\diag(\alpha_1,\dots,\alpha_n)$. With the Fr\'echet derivative $DG(A):\bH_n\to\bH_n$ of
$G$ at $A$, for $K=[K_{ij}]$ we have
\begin{align*}
{d\over dt}G(A+t\im[A,K])\Big|_{t=0}
&=DG(A)(\im[A,K])=\biggl[{G(\alpha_i)-G(\alpha_j)\over\alpha_i-\alpha_j}\biggr]_{ij}
\circ[\im(\alpha_i-\alpha_j)K_{ij}]_{ij} \\
&=\im[(G(\alpha_i)-G(\alpha_j))K_{ij}]=\im[G(A),K].
\end{align*}

(3)\enspace
By the isometry property together with the above (1) and Lemma \ref{L-1.1}\,(1) we get
\begin{align*}
K_A(H,H)&=K_{G(A)}(G'(A)H,G'(A)H)
=\Tr\hat\phi(G(A))G'(A)^2H^2 \\
&=\sum_{i=1}^k{G'(\lambda_i)^2\over\phi(G(\lambda_i),G(\lambda_i))}\Tr P_iH^2.
\end{align*}

(4)\enspace
By the isometry property together with the above (1), (2) and Lemma \ref{L-1.1}\,(2) we get
$$
K_A(H,\im[A,K])=K_{G(A)}(G'(A)H,\im[G(A),K])=0.
$$

(5)\enspace
Similarly, by Lemma \ref{L-1.1}\,(3),
\begin{align*}
K_A(\im[A,K],\im[A,K])
&=K_{G(A)}(\im[G(A),K],\im[G(A),K]) \\
&=\sum_{i,j=1}^k{G(\lambda_i)-G(\lambda_j))^2\over\phi(G(\lambda_i),G(\lambda_j))}
\Tr P_iKP_jK.
\end{align*}
\end{proof}

In particular, let $G$ be a smooth function from $(0,\infty)$ into $(0,\infty)$ such that
$G'(x)\ne0$ for all $x>0$. Let $G^{[1]}(x,y)$ be the {\it divided difference} of $G$, i.e.,
$$
G^{[1]}(x,y):=\begin{cases}
{G(x)-G(y)\over x-y} & \text{if $x\ne y$}, \\
G'(x) & \text{if $x=y$}.
\end{cases}
$$

Then, from Lemmas \ref{L-1.1} and \ref{L-1.2} we arrive at the following result.

\begin{thm}
The pull-back of the kernel metric $K^\phi$ under the mapping $D\in\bP_n\mapsto G(D)\in\bP_n$ is a kernel metric $K^\psi$ corresponding to the function
$$
\psi(x,y):={\phi(x,y)\over G^{[1]}(x,y)^2},\qquad x,y>0.
$$
\end{thm}

\section{Pull-back metrics from the Euclidean metric}
\setcounter{equation}{0}

We are concerned with the Riemannian metric $K^\phi$ related to a kernel function $\phi$
which is a power of a certain mean for two positive numbers. As in \cite{HK} a
{\it symmetric homogeneous mean} is a function $M:(0,\infty)\times(0,\infty)\to(0,\infty)$
such that for every $x,y>0$,
\begin{itemize}
\item[(1)] $M(x,y)=M(y,x)$,
\item[(2)] $M(\alpha x,\alpha y)=\alpha M(x,y)$ for all $\alpha>0$,
\item[(3)] $M(x,y)$ is non-decreasing in $x,y$,
\item[(4)] $\min\{x,y\}\le M(x,y)\le\max\{x,y\}$.
\end{itemize}
The above mean $M$ is determined by a single variable function $M(x,1)$ since
$M(x,y)=yM(x/y,1)$. The set of all symmetric homogeneous means was denoted by $\MM$ in
\cite{HK}, so in this paper we denote by $\MM_0$ the set of all smooth symmetric
homogeneous means. Here a symmetric homogeneous mean $M(x,y)$ is smooth if so is $M(x,1)$.
This means that $M(x,y)$ is smooth in $x,y>0$.

In the rest of the paper we always assume $n\ge2$ since the situation is trivial when
$n=1$. We assume that $\phi$ is a power of an $M\in\MM_0$ with degree $\theta\in\bR$,
i.e., $\phi(x,y):=M(x,y)^\theta$. The aim of this section is to determine when the
Riemannian metric $K^\phi$ derived from $M$ and $\theta$ is a pull-back of the Euclidean
metric. We are interested in this problem because the geodesic shortest path in that case
is explicitly written as the pull-back of a segment in the Euclidean space.

\begin{thm}\label{T-2.1}
Let $M\in\MM_0$, $\theta\in\bR$ with $\theta\ne0$ and $\phi(x,y):=M(x,y)^\theta$. Assume
that $F$ is a smooth function from $(0,\infty)$ into $\bR$ such that $F'(x)\ne0$ for all
$x>0$. Then the transformation $D\in\bP_n\mapsto F(D)\in\bH_n$ is isometric from
$(\bP_n,K^\phi)$ into the Euclidean manifold $(\bH_n,\|\cdot\|_\HS)$ if and only if
\begin{equation}\label{F-2.1}
F(x)=\begin{cases}
\pm{2\over2-\theta}x^{2-\theta\over2}+c & \text{if $\theta\ne0,2$}, \\
\pm\log x+c & \text{if $\theta=2$},
\end{cases}
\end{equation}
(up to a constant $c$) and
\begin{equation}\label{F-2.2}
M(x,y)=\begin{cases}
\displaystyle
\Biggl({2-\theta\over2}\cdot{x-y\over x^{2-\theta\over2}-y^{2-\theta\over2}}
\Biggr)^{2/\theta} & \text{if $\theta\ne0,2$}, \\
\displaystyle{x-y\over\log x-\log y} & \text{if $\theta=2$}.
\end{cases}
\end{equation}

Moreover, in this case, for every $A,B\in\bP_n$ a unique geodesic shortest curve from $A$
to $B$ is given by
\begin{equation}\label{F-2.3}
\gamma(t)=\begin{cases}
\Bigl((1-t)A^{2-\theta\over2}+tB^{2-\theta\over2}\Bigr)^{2\over2-\theta},
& \text{$0\le t\le 1$\qquad if $\theta\ne0,2$}, \\
\exp((1-t)\log A+t\log B), & \text{$0\le t\le1$\qquad if $\theta=2$},
\end{cases}
\end{equation}
and the geodesic distance between $A$ and $B$ is
$$
\delta_\phi(A,B)=\begin{cases}
{2\over|2-\theta|}\|A^{2-\theta\over2}-B^{2-\theta\over2}\|_\HS
& \text{if $\theta\ne0,2$}, \\
\|\log A-\log B\|_\HS & \text{if $\theta=2$}.
\end{cases}
$$
\end{thm}

\begin{proof}
Let $(a,b)$ be the range of $F$ (which must be an open interval by assumption) and
$G:=F^{-1}:(a,b)\to(0,\infty)$ be the inverse of $F$. The stated property of isometric
transformation means that the pull-back of $K^\phi$ via $G$ is the Euclidean metric on
the submanifold $\bH_n(a,b)$ of $\bH_n$. From (3)--(5) of Lemma \ref{L-1.2} one can easily
see that this property is equivalent to that the following two conditions hold:
$$
{G'(t)^2\over G(t)^\theta}=1,\qquad t\in(a,b),
$$
$$
{(G(s)-G(t))^2\over\phi(G(s),G(t))}=(s-t)^2,\qquad s,t\in(a,b).
$$
It is obvious that the above two are respectively equivalent to the following:
\begin{equation}\label{F-2.4}
F'(x)^2=x^{-\theta},\qquad x>0,
\end{equation}
\begin{equation}\label{F-2.5}
{(x-y)^2\over\phi(x,y)}=(F(x)-F(y))^2,\qquad x,y>0.
\end{equation}
The differential equation \eqref{F-2.4} determines $F$ as \eqref{F-2.1}, and this together
with \eqref{F-2.5} determines $M$ as \eqref{F-2.2}.

The rest of the theorem immediately follows from the isometric transformation via $F$
in \eqref{F-2.1}. One may just note that the segment joining $H,K\in\bH_n$ is a unique
shortest path between $H$ and $K$ in the Euclidean manifold $(\bH_n,\|\cdot\|_\HS)$.
\end{proof}

In the following we present a bit more direct proof of Theorem \ref{T-2.1}. Formula
\eqref{F-2.7} below will be also useful in our discussions in the rest of the paper. Let
$F$ and $G:=F^{-1}$ be as above. For each $C^1$ curve $\gamma:[0,1]\to\bP_n$ we make a
change of variable $\xi(t):=F(\gamma(t))$, hence $\gamma(t)=G(\xi(t))$. We then have
$$
K_{\gamma(t)}^\phi(\gamma'(t),\gamma'(t))
=\|\phi(\bL_{\gamma(t)},\bR_{\gamma(t)})^{-1/2}\gamma'(t)\|_\HS^2
$$
and
$$
\gamma'(t)=DG(\xi(t))(\xi'(t)),
$$
where $DG(\xi(t)):\bH_n\to\bH_n$ is the Fr\'echet derivative of $G$ at $\xi(t)$. Under the
diagonalization $\xi(t)=U\diag(\lambda_1,\dots,\lambda_n)U^*$ for each fixed $t\in[0,1]$,
thanks to the differential formula (see \cite{Bh})
\begin{equation}\label{F-2.6}
DG(\xi(t))(\xi'(t))
=U\bigl(\bigl[G^{[1]}(\lambda_i,\lambda_j)\bigr]_{ij}\circ(U^*\xi'(t)U)\bigr)U^*
=G^{[1]}(\bL_{\xi(t)},\bR_{\xi(t)})\xi'(t)
\end{equation}
as well as \eqref{F-1.2}, we obtain
\begin{align}\label{F-2.7}
\phi(\bL_{\gamma(t)},\bR_{\gamma(t)})^{-1/2}\gamma'(t)
&=\phi(\bL_{G(\xi(t))},\bR_{G(\xi(t))})^{-1/2}
G^{[1]}(\bL_{\xi(t)},\bR_{\xi(t)})\xi'(t) \nonumber\\
&=U\Biggl(\Biggl[{G^{[1]}(\lambda_i,\lambda_j)\over
\sqrt{\phi(G(\lambda_i),G(\lambda_j))}}\Biggr]_{ij}\circ(U^*\xi'(t)U)\Biggr)U^*.
\end{align}
Hence we see that the metric $K^\phi$ on $\bP_n$ is the pull-back of the Euclidean metric
on $\bH_n(a,b)$ via $F$ if and only if
\begin{equation}\label{F-2.8}
{G^{[1]}(s,t)\over\sqrt{\phi(G(s),G(t))}}=\pm1
\end{equation}
for all $s,t\in(a,b)$, where the right-hand side of \eqref{F-2.8} is $1$ or $-1$ according
to $G$ being increasing or decreasing. Since $\phi(x,x)=x^\theta$, \eqref{F-2.8} for $s=t$
yields the differential equation
$$
G'(t)=\pm G(t)^{\theta/2},\qquad t\in(a,b).
$$
This is equivalently written as $F'(x)=\pm x^{-\theta/2}$, $x>0$, which is solved as
\eqref{F-2.1}. From \eqref{F-2.1} and \eqref{F-2.8} we obtain \eqref{F-2.2}. Thus we have
proved Theorem \ref{T-2.1} again. Note that one can even more simply prove the theorem by
appealing to
$$
F^{[1]}(x,y)=\pm{1\over\sqrt{\phi(x,y)}}.
$$

For $\theta\in\bR$, $\theta\ne0$, we write $M_\theta(x,y)$ for $M(x,y)$ given in
\eqref{F-2.2} and $\phi_\theta(x,y)$ for $M_\theta(x,y)^\theta$. The family of means
$M_\theta$ interpolates the following typical means:
\begin{align}
M_{-2}(x,y)&=M_\mathrm{A}(x,y):={x+y\over2}\quad\mbox{(arithmetic mean)},
\label{F-2.9}\\
M_1(x,y)&=M_{\sqrt{\phantom{a}}}(x,y):=\biggl({\sqrt x+\sqrt y\over2}\biggr)^2
\quad\mbox{(root mean)}, \label{F-2.10}\\
M_2(x,y)&=M_\mathrm{L}(x,y):={x-y\over\log x-\log y}\quad\mbox{(logarithmic mean)},
\label{F-2.11}\\
M_4(x,y)&=M_\mathrm{G}(x,y):=\sqrt{xy}\quad\mbox{(geometric mean)}.
\label{F-2.12}
\end{align}
Furthermore, we may define $M_0(x,y)$ by taking the limit
\begin{equation}\label{F-2.13}
M_0(x,y):=\lim_{\theta\to0}M_\theta(x,y)
={1\over e}\biggl({x^x\over y^y}\biggr)^{1/(x-y)}\quad\mbox{(identric mean)},
\end{equation}
and $\phi_0(x,y)\equiv1$. Note also that
$$
{x-y\over\log x-\log y}=\lim_{\theta\to2}
{2-\theta\over2}\cdot{x-y\over x^{2-\theta\over2}-y^{2-\theta\over2}}.
$$

As mentioned in Introduction, monotone metrics (\cite{Pe3}) are among particularly
important class of Riemannian metrics. Those are the kernel metrics $K^\phi$ in the case
where $\theta=1$ and $M(x,1)$ is operator monotone. In the case $\theta=1$, the theorem
says that the metric corresponding to the root mean $M_{\sqrt{\phantom{a}}}$ (that is a
special case of binomial means \cite{HK}), called the {\it Wigner-Yanase metric}, is a
unique monotone metric which is a pull-back of the Euclidean metric. This was in fact
proved by Gibilisco and Isola \cite{GI} in a slightly different approach. Other famous
monotone metrics are the {\it Bogoliubov metric} (also called the {\it Kubo-Mori metric})
corresponding to the logarithmic mean $M_\mathrm{L}$ and the {\it Bures-Uhlmann metric}
corresponding to the arithmetic mean $M_\mathrm{A}$.

In this way, we have found a one-parameter family $M_\theta\in\MM_0$, $\theta\in\bR$,
given in \eqref{F-2.2} and \eqref{F-2.13}. It is remarkable that this is a rather familiar
family of means introduced in \cite{St} with a different parametrization and called
{\it Stolarsky means} in \cite[\S2.6]{BK}. A monotonicity property of the family was
proved in \cite{St}, which we state in the next lemma for the convenience of references.

\begin{lemma}\label{L-2.2}{\rm(\cite{St})}\enspace
For every $x,y>0$ with $x\ne y$, $M_\theta(x,y)$ is strictly decreasing in $\theta\in\bR$.
Furthermore, $\lim_{\theta\to-\infty}M_\theta(x,y)=\max\{x,y\}$ and
$\lim_{\theta\to\infty}M_\theta(x,y)=\min\{x,y\}$.
\end{lemma}

Next we are concerned with the relation among the metrics $K^\phi$ under the reflection
map $A\mapsto A^{-1}$.

\begin{prop}\label{P-2.3}
Let $M^{(1)},M^{(2)}\in\MM_0$, $\theta_1,\theta_2\in\bR$ and
$\phi^{(k)}(x,y):=M^{(k)}(x,y)^{\theta_k}$, $k=1,2$. Then the Riemannian manifolds
$(\bP_n,K^{\phi^{(1)}})$ and $(\bP_n,K^{\phi^{(2)}})$ are isometric under the reflection
$A\mapsto A^{-1}$ on $\bP_n$ if and only if
$$
\theta_1+\theta_2=4\quad\mbox{and}\quad
\biggl({M^{(1)}(x,y)\over\sqrt{xy}}\biggr)^{\theta_1}
=\biggl({M^{(2)}(x,y)\over\sqrt{xy}}\biggr)^{\theta_2},\quad x,y>0.
$$
In particular, if $\phi(x,y)=M(x,y)^2$ with an arbitrary $M\in\MM_0$, then
$A\mapsto A^{-1}$ is an isometric transformation on $(\bP_n,K^\phi)$. Moreover, for every
$\theta\in\bR$, $(\bP_n,K^{\phi_\theta})$ and $(\bP_n,K^{\phi_{4-\theta}})$ are isometric
under $A\mapsto A^{-1}$.
\end{prop}

\begin{proof}
If $\gamma$ is a $C^1$ curve in $\bP_n$, then we have
$$
\phi^{(1)}(\bL_{\gamma(t)^{-1}},\bR_{\gamma(t)^{-1}})^{-1/2}
\biggl({d\over dt}\gamma(t)^{-1}\biggr)
=\phi^{(1)}(\bL_{\gamma(t)}^{-1},\bR_{\gamma(t)}^{-1})^{-1/2}
\bL_{\gamma(t)}^{-1}\bR_{\gamma(t)}^{-1}\gamma'(t).
$$
Hence $A\mapsto A^{-1}$ gives an isometry between $(\bP_n,K^{\phi_1})$ and
$(\bP_n,K^{\phi_2})$ if and only if
$$
\|\phi^{(1)}(\bL_D^{-1},\bR_D^{-1})^{-1/2}\bL_D^{-1}\bR_D^{-1}H\|_\HS
=\|\phi^{(2)}(\bL_D,\bR_D)^{-1/2}H\|_\HS
$$
for all $D\in\bP_n$ and $H\in\bH_n$. We may assume that $D$ is diagonal. For
$D=\diag(\lambda_1,\dots,\lambda_n)$ the above equality is written as
$$
\left\|\left[{1\over\sqrt{\phi^{(1)}(\lambda^{-1},\lambda_j^{-1})}\,\lambda_i\lambda_j}
\right]_{ij}\circ H\right\|_\HS
=\left\|\left[{1\over\sqrt{\phi^{(2)}(\lambda_i,\lambda_j)}}\right]_{ij}\circ H\right\|_\HS.
$$
This hold for all $H\in\bH_n$ if and only if
$$
\phi^{(1)}(x^{-1},y^{-1})x^2y^2=\phi^{(2)}(x,y),\qquad x,y>0,
$$
that is,
$$
M^{(1)}(x^{-1},y^{-1})^{\theta_1}x^2y^2=M^{(2)}(x,y)^{\theta_2},\qquad x,y>0.
$$
Letting $x=y$ implies $x^{4-\theta_1}=x^{\theta_2}$. Hence $\theta_1+\theta_2=4$ must hold
and the above condition for $M^{(1)}$ and $M^{(2)}$ is rewritten as
$$
\biggl({M^{(1)}(x,y)\over\sqrt{xy}}\biggr)^{\theta_1}
=\biggl({M^{(2)}(x,y)\over\sqrt{xy}}\biggr)^{\theta_2},\qquad x,y>0.
$$
Since this is obviously satisfied for $\theta_1=\theta_2=2$ and $M^{(1)}=M^{(2)}$, the
second assertion follows. A simple computation with \eqref{F-2.2} gives the last
assertion.
\end{proof}

\begin{remark}\label{R-2.4}{\rm
The latter assertions of Proposition \ref{P-2.3} can be extended as follows: For every
$\theta,\theta'\in\bR\setminus\{2\}$ the Riemannian manifolds $(\bP_n,K^{\phi_\theta})$
and $(\bP_n,K^{\phi_{\theta'}})$ are isometric under the diffeomorphism
$$
A\mapsto\bigg|{2-\theta\over2-\theta'}\bigg|^{2\over2-\theta}A^{2-\theta'\over2-\theta},
$$
and for every $\alpha\in\bR\setminus\{0\}$, $A\mapsto A^\alpha$ is an isometric
transformation on $(\bP_n,K^{\phi_2})$.
}\end{remark}

An interesting problem concerning the family $M_\theta$ is to determine the range of
$\theta$ for which $M_\theta$ is an operator monotone mean, i.e., $M_\theta(x,1)$ is
an operator monotone function on $(0,\infty)$. The cases $\theta=-2,1,2$ and $4$ are among
typical operator monotone functions as listed in \eqref{F-2.9}--\eqref{F-2.12}.
The problem has been settled by Kosaki \cite{Ko2} in such a way that $M_\theta(x,1)$ is
operator monotone if and only if $-2\le\theta\le6$.

We give the next lemma on $M_\theta$ for later use.

\begin{lemma}\label{L-2.5}
Let $M_\mathrm{H}(x,y):=2xy/(x+y)$, the harmonic mean. Then
$M_{10}(x,1)>H_\mathrm{H}(x,1)$ for all $x>0$ with $x\ne1$. For every $\theta>10$,
$M_\theta(x,1)<M_\mathrm{H}(x,1)$ if $x$ $(\ne1)$ is sufficiently near $1$.
\end{lemma}

\begin{proof}
The proof of the first assertion is elementary and omitted. To prove the second, let
$\theta>10$ and $\alpha:=(\theta-2)/2>4$. Direct computations show
\begin{align*}
M_\theta(x,1)^{\alpha+1}&=\alpha{x^{\alpha+1}-x^\alpha\over x^\alpha-1} \\
&=1+{\alpha+1\over2}(x-1)+{(\alpha+1)(\alpha-1)\over12}(x-1)^2+o((x-1)^2), \\
M_\mathrm{H}(x,1)^{\alpha+1}&=\biggl({2x\over x+1}\biggr)^{\alpha+1} \\
&=1+{\alpha+1\over2}(x-1)+{(\alpha+1)(\alpha-2)\over8}(x-1)^2+o((x-1)^2),
\end{align*}
which give the desired assertion.
\end{proof}

\section{The degree $2$ case}
\setcounter{equation}{0}

A Riemannian manifold said to be {\it complete} if the distance induced from the Riemannian
metric is complete. It is a general fact in Riemannian geometry that a geodesic shortest
curve joining any two points exists in a complete Riemannian manifold. The next theorem
shows that the Riemannian manifold $(\bP_n,K^\phi)$ treated in Section 2 is never complete
except the case of degree $\theta=2$.

\begin{thm}\label{T-3.1}
Let $M\in\MM_0$, $\theta\in\bR$ and $\phi(x,y):=M(x,y)^\theta$. Then the Riemannian
manifold $(\bP_n,K^\phi)$ is complete if and only if $\theta=2$. Hence, when $\theta=2$
(and $M\in\MM_0$ is arbitrary), for any $A,B\in\bP_n$ there is a geodesic shortest curve
joining $A,B$ in $(\bP_n,K^\phi)$.
\end{thm}

\begin{proof}
First assume $\theta\ne2$. The proof of the non-completeness of $(\bP_n,K^\phi)$ is easy.
Let $\gamma(t):=tI$ for $t>0$, where $I$ is the $n\times n$ identity matrix. Since
$$
\|\phi(\bL_{\gamma(t)},\bR_{\gamma(t)})^{-1/2}\gamma'(t)\|_\HS
=\|M(t,t)^{-\theta/2}I\|_\HS=t^{-\theta/2}\sqrt n,
$$
we have
$$
\int_0^1\sqrt{K_{\gamma(t)}^\phi(\gamma'(t),\gamma'(t))}\,dt<+\infty
\quad\mbox{if $\theta<2$},
$$
$$
\int_1^\infty\sqrt{K_{\gamma(t)}^\phi(\gamma'(t),\gamma'(t))}\,dt<+\infty
\quad\mbox{if $\theta>2$}.
$$
Hence, if we define $A_k:={1\over k}I_n$ if $\theta<2$ and $A_k:=kI_n$ if $\theta>2$, then
it follows that $\{A_k\}_{k=1}^\infty$ is Cauchy with respect to the geodesic distance
$\delta_\phi$. But it is clear that $\{A_k\}$ does not converge in $(\bP_n,K^\phi)$.

Next assume $\theta=2$, and prove that $(\bP_n,K^\phi)$ is complete. To do so, we need a
lemma.

\begin{lemma}\label{L-3.2}
If $M\in\MM_0$ and $\phi(x,y):=M(x,y)^2$, then $\delta_\phi(A,I)=\|\log A\|_\HS$ for every
$A\in\bP_n$.
\end{lemma}

\begin{proof}
We may assume that $A$ is diagonal. Let $\gamma:[0,1]\to\bP_n$ be a $C^1$ curve from $A$
to $I$, and diagonalize $\gamma(t)$, $0\le t\le1$, so that
$$
\gamma(t)=U(t)\diag(\lambda_1(t),\dots,\lambda_n(t))U(t)^*
$$
with $\lambda_1(t)\le\dots\le\lambda_n(t)$ and unitary matrices $U(t)$. Here one can fix
$U(t)$, $0\le t\le1$, so that $\lambda_1(t),\dots,\lambda_n(t)$ and $U(t)$ are $C^1$
except branching points of $\lambda_1(t),\dots,\lambda_n(t)$ (see \cite{Ka} for example).
Note that the set of branching points is at most countable. Hence, for each $t$ except
such branching points, we have
\begin{align*}
\gamma'(t)=U(t)\diag(\lambda_1'(t),\dots,\lambda_n'(t))U(t)^*
&+U'(t)\diag(\lambda_1(t),\dots,\lambda_n(t))U(t)^* \\
&+U(t)\diag(\lambda_1(t),\dots,\lambda_n(t))U'(t)^*
\end{align*}
so that
\begin{align*}
U(t)^*\gamma'(t)U(t)=\diag(\lambda_1'(t),\dots,\lambda_n'(t))
&+U(t)^*U'(t)\diag(\lambda_1(t),\dots,\lambda_n(t)) \\
&+\diag(\lambda_1(t),\dots,\lambda_n(t))U'(t)^*U(t).
\end{align*}
Since $U(t)^*U(t)=I$ yields that $U'(t)^*U(t)+U(t)^*U'(t)=O$, the diagonal entries of
$U(t)^*\gamma'(t)U(t)$ are $\lambda_1'(t),\dots,\lambda_n'(t)$. Hence we get
\begin{align*}
&\|\phi(\bL_{\gamma(t)},\bR_{\gamma(t)})^{-1/2}\gamma'(t)\|_\HS \\
&\qquad=\Bigg\|\biggl[{1\over M(\lambda_i(t),\lambda_j(t))}\biggr]_{ij}
\circ(U(t)^*\gamma'(t)U(t))\Bigg\|_\HS
\ge\sqrt{\sum_{i=1}^n\biggl({\lambda_i'(t)\over\lambda_i(t)}\biggr)^2}
\end{align*}
for all $t$ except a countable set. Since
$\xi(t):=\diag(\log\lambda_1(t),\dots,\log\lambda_n(t))$ is a curve (continuous in
$0\le t\le1$ and $C^1$ except a countable set as mentioned above) from $\log A$ to $O$,
we get
$$
L_\phi(\gamma)\ge\int_0^1\|\xi'(t)\|_\HS\,dt\ge\|\log A\|_\HS.
$$
Furthermore, if $A=\diag(\lambda_1,\dots,\lambda_n)$ and
$\gamma_0(t):=A^{1-t}=\diag(\lambda_1^{1-t},\dots,\lambda_n^{1-t})$ for $0\le t\le1$, then
one can easily compute
$$
L_\phi(\gamma_0)=\sqrt{\sum_{i=1}^n(\log\lambda_i)^2}
=\|\log A\|_\HS,
$$
implying $\delta_\phi(A,I)=\|\log A\|_\HS$.
\end{proof}

\noindent
{\it Proof of Theorem \ref{T-3.1} (continued)}.\enspace
Let $\{A_k\}$ be a $\delta_\phi$-Cauchy sequence in $\bP_n$. Since
$|\delta_\phi(A_k,I)-\delta_\phi(A_l,I)|\le\delta_\phi(A_k,A_l)\to0$ as $k,l\to\infty$, it
follows from Lemma \ref{L-3.2} that $\delta_\phi(A_k,I)=\|\log A_k\|_\HS$ is a bounded
sequence and so $\sup_k\|\log A_k\|_\infty<+\infty$ ($\|\cdot\|_\infty$ being the operator
norm). Hence there is an $\eps>0$ such that $\eps I\le A_k\le\eps^{-1}I$ for all $k$. By
compactness we can choose a subsequence $\{A_{k_m}\}$ of $\{A_k\}$ such that
$\|A_{k_m}-A\|_\infty\to0$ for some $A\in\bP_n$ with $\eps I\le A\le\eps^{-1}I$. Here we
may assume that $\{A_k\}$ itself converges to $A$ in operator norm. Then we have
$\|\log A_k-\log A\|_\infty\to0$ and so $\|\log A_k-\log A\|_\HS\to0$. Define
$\xi_k(t):=(1-t)\log A_k+t\log A$ and $\gamma_k(t)=e^{\xi_k(t)}$ for $0\le t\le1$. For
each fixed $t\in[0,1]$ diagonalize $\xi_k(t)$ as $\xi_k(t)=V\diag(\mu_1,\dots,\mu_n)V^*$
with a unitary $V$. By \eqref{F-2.7} we get
\begin{align*}
\|M(\gamma_k(t),\gamma_k(t))^{-1}\gamma_k'(t)\|_\HS
&=\bigg\|\biggl[{1\over M(e^{\mu_i},e^{\mu_j})}\cdot
{e^{\mu_i}-e^{\mu_j}\over\mu_i-\mu_j}\biggr]\circ(V^*\xi_k'(t)V)\bigg\|_\HS \\
&=\bigg\|\biggl[{M_\mathrm{L}(e^{\mu_i},e^{\mu_j})\over
M(e^{\mu_i},e^{\mu_j})}\biggr]\circ(V^*\xi_k'(t)V)\bigg\|_\HS,
\end{align*}
where $M_\mathrm{L}(x,y)$ is the logarithmic mean. Since $\eps\le e^{\mu_i}\le\eps^{-1}$
for all $1\le i\le n$, it follows that
$$
{M_\mathrm{L}(e^{\mu_i},e^{\mu_j})\over M(e^{\mu_i},e^{\mu_j})}
\le{M_\mathrm{L}(\eps^{-1},\eps^{-1})\over M(\eps,\eps)}=\eps^{-2}
\quad\mbox{for all $i,j$}.
$$
Therefore, for every $0\le t\le1$,
$$
\|M(\gamma_k(t),\gamma_k(t))^{-1}\gamma_k'(t)\|_\HS
\le\eps^{-2}\|\xi_k'(t)\|_\HS=\eps^{-2}\|\log A_k-\log A\|_\HS
$$
so that
$$
\delta_\phi(A_k,A)\le L_\phi(\gamma_k)
\le\eps^{-2}\|\log A_k-\log A\|_\HS\longrightarrow0
\quad\mbox{as $k\to\infty$}.
$$
Hence the result follows.
\end{proof}

Let $\phi_\mathrm{G}(x,y)$ denote the degree $2$ power of the geometric mean, i.e.,
$\phi_\mathrm{G}(x,y):=M_\mathrm{G}(x,y)^2=xy$. The metric $K^{\phi_\mathrm{G}}$ induced
from $\phi_\mathrm{G}$ is the Fisher-Rao metric $g$ mentioned in Introduction. The
completeness of the Riemannian manifold $(\bP_n,K^{\phi_\mathrm{G}})$ was shown in
\cite{BH}. Now we define a one-parameter family of kernel functions
\begin{equation}\label{F-3.1}
N_\alpha(x,y):=\alpha(xy)^{\alpha/2}{x-y\over x^\alpha-y^\alpha},
\quad x,y>0,\quad\alpha\in\bR,
\end{equation}
where $N_0(x,y)$ is understood as
$$
N_0(x,y):=\lim_{\alpha\to0}N_\alpha(x,y)={x-y\over\log x-\log y}
\quad\mbox{(logarithmic mean)}.
$$
We have $N_1(x,y)=\sqrt{xy}$ (geometric mean) and $N_2(x,y)=2xy/(x+y)$ (harmonic mean).
Note that $N_\alpha(x,y)$ is symmetric and homogeneous in the sense of (1) and (2) at the
beginning of Section 2 and $N_{-\alpha}(x,y)=N_\alpha(x,y)$. When $\alpha>2$, $N_\alpha$
does not belong to $\MM_0$ since $N_\alpha(x,1)\to0$ as $x\to\infty$. When $0<\alpha\le2$,
one can easily see by elementary calculus that $N_\alpha(x,1)$ is increasing in $x>0$ and
$1\le N_\alpha(x,1)\le x$ for all $x\ge1$. It is also not difficult to see that
$N_\alpha(x,y)$ is strictly decreasing in $\alpha>0$ for each $x,y>0$ with $x\ne y$. Thus
$\{N_\alpha\}_{0\le\alpha\le2}$ is a family of means in $\MM_0$ interpolating the
logarithmic and the harmonic means.

We determine when our Riemannian metric $K^\phi$ is a pull-back of $K^{\phi_\mathrm{G}}$
up to a multiple constant, and moreover extend \eqref{F-0.4} and \eqref{F-0.5} for
$g=K^{\phi_\mathrm{G}}$ to the family of metrics induced from the above $N_\alpha$.

\begin{thm}\label{T-3.3}
Let $M\in\MM_0$, $\theta\in\bR$ and $\phi(x,y):=M(x,y)^\theta$. Let $\alpha>0$. Assume
that $F$ is a smooth function from $(0,\infty)$ into itself such that $F'(x)\ne0$ for all
$x>0$. Then the transformation $D\in\bP_n\mapsto F(D)\in\bP_n$ is isometric from
$(\bP_n,\alpha K^\phi)$ into $(\bP_n,K^{\phi_\mathrm{G}})$ if and only if $\theta=2$,
$\alpha\le2$, $F(x)=cx^\alpha$ (up to a constant $c>0$) and $M=N_\alpha$.

In the above case, for every $A,B\in\bP_n$ there exists a unique geodesic shortest curve
in $(\bP_n,K^\phi)$ from $A$ to $B$ given by
$$
\gamma(t):=(A^\alpha\,\#_t\,B^\alpha)^{1/\alpha}\ \Bigl(
=\bigl(A^{\alpha/2}(A^{-\alpha/2}B^\alpha A^{-\alpha/2})^tA^{\alpha/2}\bigr)^{1/\alpha}
\Bigr)
$$
and moreover
$$
\delta_\phi(A,B)=\|\log(A^{-\alpha/2}B^\alpha A^{-\alpha/2})^{1/\alpha}\|_\HS.
$$
\end{thm}

\begin{proof}
For any $C^1$ curve $\gamma$ in $\bP_n$ let $\xi(t):=F(\gamma(t))$. Under the
diagonalization $\gamma(t)=U\diag(\lambda_1,\dots,\lambda_n)U^*$ for each fixed
$t\in[0,1]$ we have by \eqref{F-1.2} and \eqref{F-2.7}
\begin{align*}
\|\phi(\bL_{\gamma(t)},\bR_{\gamma(t)})\gamma'(t)\|_\HS
&=\Bigg\|\Biggl[{1\over\sqrt{\phi(\lambda_i,\lambda_j)}}\Biggr]_{ij}
\circ(U^*\gamma'(t)U)\Bigg\|_\HS, \\
\|\phi_\mathrm{G}(\bL_{\xi(t)},\bR_{\xi(t)})\xi'(t)\|_\HS
&=\Bigg\|\Biggl[{F^{[1]}(\lambda_i,\lambda_j)\over
\sqrt{F(\lambda_i)F(\lambda_j)}}\Biggr]_{ij}
\circ(U^*\gamma'(t)U)\Bigg\|_\HS.
\end{align*}
Hence the isometry property stated in the theorem implies that
\begin{equation}\label{F-3.2}
{F^{[1]}(x,y)\over\sqrt{F(x)F(y)}}={\alpha\over\sqrt{\phi(x,y)}},\qquad x,y>0.
\end{equation}
When $x=y$ this yields
\begin{equation}\label{F-3.3}
{F'(x)\over F(x)}=\alpha x^{-\theta/2},\qquad x>0.
\end{equation}
Suppose $\theta\ne2$. Then \eqref{F-3.3} is solved as
$$
F(x)=c\exp\biggl({2\alpha\over2-\theta}\,x^{2-\theta\over2}\biggr)
$$
with a constant $c>0$. By this and \eqref{F-3.2}, $\phi(x,1)$ is written as
$$
\phi(x,1)=\alpha^2e^\beta\,{e^{\beta x^r}(x-1)^2\over(e^{\beta x^r}-e^\beta)^2}
\quad\mbox{with}\ r:={2-\theta\over2},\ \beta:={2\alpha\over2-\theta}.
$$
If $0\le\theta<2$, then $\phi(x,1)\to0$ as $x\to\infty$. But this is inconsistent with
$\phi(x,1)=M(x,1)^\theta\ge1$ for $x\ge1$. If $\theta<0$, then $x^{-\theta}\phi(x,1)\to0$
as $x\to\infty$, which is inconsistent with $\phi(x,1)=M(x,1)^\theta\ge x^\theta$ for
$x\ge1$. If $\theta>2$, then $x^{-\theta}\phi(x,1)\to0$ as $x\to0$, which is also
inconsistent with $\phi(x,1)\ge x^\theta$ for $0<x\le1$. Hence all the cases except
$\theta=2$ are excluded. When $\theta=2$,  the solution of \eqref{F-3.3} is
$F(x)=cx^\alpha$ with a constant $c>0$. This and \eqref{F-3.2} determine $M$ as
$M=N_\alpha$. Then $\alpha$ obeys the restriction $\alpha\le2$ as shown before the
theorem. It is immediate to see that the isometry property actually holds if $\theta$,
$\alpha$, $F$ and $M$ are as stated in the theorem.

When $\alpha K^\phi$ is a pull-back of $K^{\phi_\mathrm{G}}$ as above, a geodesic shortest
curve in $(\bP_n,K^\phi)$ joining each $A,B\in\bP_n$ is uniquely determined as the image
under $A\mapsto A^{1/\alpha}$ of that in $(\bP_n,K^{\phi_G})$ joining $A^\alpha,B^\alpha$.
Thanks to \eqref{F-0.4} its explicit form is
$$
\gamma(t):=(A^\alpha\,\#_t\,B^\alpha)^{1/\alpha},\qquad0\le t\le1.
$$
Furthermore, thanks to \eqref{F-0.5} it is also immediate to get
\begin{align*}
\delta_\phi(A,B)&={1\over\alpha}\delta_{\phi_\mathrm{G}}(A^\alpha,B^\alpha)
={1\over\alpha}\|\log(A^{-\alpha/2}B^\alpha A^{-\alpha/2})\|_\HS \\
&=\|\log(A^{-\alpha/2}B^\alpha A^{-\alpha/2})^{1/\alpha}\|_\HS,
\end{align*}
as required.
\end{proof}

It is desirable to prove the uniqueness of geodesic shortest curves for all metrics
treated in Theorem \ref{T-3.1} in the degree $2$ case.

We write $\psi_\alpha$ for $\phi$ arising in Theorem \ref{T-3.3}, i.e.,
$\psi_\alpha(x,y):=N_\alpha(x,y)^2$ for $0<\alpha\le2$. It is worth noting that the
geodesic shortest path and its distance in $(\bP_n,K^{\psi_\alpha})$ converge as
$\alpha\searrow0$ to those in $(\bP_n,K^{\phi_\mathrm{L}})$ where
$\phi_\mathrm{L}(x,y):=M_\mathrm{L}(x,y)^2$, the degree 2 power of the logarithmic mean.
Namely, we have
$$
\lim_{\alpha\searrow0}(A^\alpha\,\#_t\,B^\alpha)^{1/\alpha}
=\exp((1-t)\log A+t\log B),\qquad0\le t\le1,
$$
$$
\lim_{\alpha\searrow0}\|\log(A^{-\alpha/2}B^\alpha A^{-\alpha/2})^{1/\alpha}
\|_\HS=\|\log A-\log B\|_\HS
$$
(see the $\theta=2$ case of Theorem \ref{T-2.1}). In fact, the latter follows from a
version of the Lie-Trotter formula
$$
\lim_{\alpha\to0}(A^{-\alpha/2}B^\alpha A^{-\alpha/2})^{1/\alpha}
=\exp(-\log A+\log B)
$$
and the former is its modification (see \cite[Lemma 3.3]{HP}). It is also worthwhile to
note that $\|\log(A^{-\alpha/2}B^\alpha A^{-\alpha/2})^{1/\alpha}\|_\HS$ is increasing in
$\alpha>0$ due to Araki's log-majorization \cite{Ar} (see also \cite{AH}). Hence
$\delta_{\psi_\alpha}(A,B)$ decreases to $\delta_{\phi_\mathrm{L}}(A,B)$ as
$\alpha\searrow0$ while $\psi_\alpha(x,y)$ increases to $\phi_\mathrm{L}(x,y)$ as
$\alpha\searrow0$. In fact, this kind of comparison property is true in general as we will
see in the following sections.

When $A_1,\dots,A_k\in\bP_n$, since the arithmetic mean ${1\over k}\sum_{j=1}^kA_j$ is the
unique minimizer of $A\in\bP_n\mapsto\sum_{j=1}^k\|A-A_j\|_\HS^2$, it is immediate from
Theorem \ref{T-2.1} that a certain power mean
$\Bigl({1\over k}\sum_{j=1}^kA_j^{2-\theta\over2}\Bigr)^{2\over2-\theta}$
(understood as $\exp\bigl({1\over k}\sum_{j=1}^n\log A_j\bigr)$ if $\theta=2$) is
determined as a unique minimizer of
$A\in\bP_n\mapsto\sum_{j=1}^k\delta_{\phi_\theta}^2(A,A_j)$. Let $G(A_1,\dots,A_k)$ be the
``geometric mean" introduced in \cite{BH,Bh3}, i.e., the unique minimizer of
$A\mapsto\sum_{j=1}^k\delta_{M_\mathrm{G}^2}^2(A,A_j)$. It is also immediately seen from
Theorem \ref{T-3.3} that $G(A_1^\alpha,\dots,A_k^\alpha)^{1/\alpha}$ is a unique
minimizer of $A\mapsto\sum_{j=1}^k\delta_{\psi_\alpha}^2(A,A_j)$, which is regarded as a
$k$-variable extension of $(A^\alpha\,\#\,B^\alpha)^{1/\alpha}$.

\section{Comparison property}
\setcounter{equation}{0}

The aim of this section is to compare the geodesic distances for different Riemannian
metrics related to means in $\MM_0$. A general result of this kind is the following:

\begin{thm}\label{T-4.1}
Let $M^{(1)},M^{(2)}\in\MM_0$, $\theta_1,\theta_2\in\bR$ and
$\phi^{(k)}(x,y):=M^{(k)}(x,y)^{\theta_k}$, $k=1,2$. Then the
following conditions are equivalent:
\begin{itemize}
\item[\rm(i)] $\phi^{(1)}(x,y)\le\phi^{(2)}(x,y)$ for all $x,y>0$;
\item[\rm(ii)] $\theta_1=\theta_2=0$, or $\theta_1=\theta_2>0$ and
$M^{(1)}(x,1)\le M^{(2)}(x,1)$ for all $x>0$, or $\theta_1=\theta_2<0$ and
$M^{(1)}(x,1)\ge M^{(2)}(x,1)$ for all $x>0$;
\item[\rm(iii)] $L_{\phi^{(1)}}(\gamma)\ge L_{\phi^{(2)}}(\gamma)$ for all $C^1$ curve in
$\bP_n$;
\item[\rm(iv)] $\delta_{\phi^{(1)}}(A,B)\ge\delta_{\phi^{(2)}}(A,B)$ for all $A,B\in\bP_n$.
\end{itemize}
\end{thm}

The next lemma is useful to prove the theorem while it is meaningful by itself.

\begin{lemma}\label{L-4.2}
Let $M\in\MM_0$, $\theta\in\bR$ and $\phi(x,y):=M(x,y)^\theta$. Then for every $D\in\bP_n$
and $H\in\bH_n$,
$$
\lim_{\eps\searrow0}{\delta_\phi(D,D+\eps H)\over\eps}
=\|\phi(\bL_D,\bR_D)^{-1/2}H\|_\HS.
$$
\end{lemma}

\begin{proof}
We may assume that $D=\diag(\lambda_1,\dots,\lambda_n)$. Notice that
$$
\phi(\bL_D,\bR_D)^{-1}H=\bigl[\phi(\lambda_i,\lambda_j)^{-1}\bigr]_{ij}\circ H
=\bigl[M(\lambda_i,\lambda_j)^{-\theta}\bigr]_{ij}\circ H
$$
so that
$$
\phi(\bL_D,\bR_D)^{-1}\ge\biggl(\min_{1\le i\le n}\lambda_i^{-\theta}\biggr)\bI
\quad\mbox{on\ \ $(\bM_n,\<\cdot,\cdot\>_\HS)$},
$$
where $\bI$ is the identity operator on $\bM_n$. For each $\rho>0$ with
$\rho<\min_i\lambda_i^{-\theta}$, since $A\in\bP_n\mapsto\phi(\bL_A,\bR_A)$ is continuous,
there exists an $r_1>0$ such that if $A\in\bP_n$ and $\|A-D\|_\HS<r_1$ then
$$
\|\phi(\bL_A,\bR_A)^{-1}-\phi(\bL_D,\bR_D)^{-1}\|_\infty<\rho,
$$
where $\|\cdot\|_\infty$ denotes the operator norm for operators on
$(\bM_n,\<\cdot,\cdot\>_\HS)$. Furthermore, since $\delta_\phi$ and $\|\cdot\|_\HS$ define
the same topology on $\bP_n$ (see \cite[Chapter IV, Proposition 3.5]{KN}), there exists an
$r_0>0$ such that if $A\in\bP_n$ and $\delta_\phi(A,D)<r_0$ then $\|A-D\|_\HS<r_1$.

Now let $H\in\bH_n$ and $\eps>0$ be sufficiently small so that
$\delta_\phi(D,D+\eps H)<r_0$ and $\eps\|H\|_\HS<r_1$. Let $\gamma:[0,1]\to\bP_n$ be any
$C^1$ curve from $D$ to $D+\eps H$ such that $L_\phi(\gamma)<r_0$. Since
$\delta_\phi(\gamma(t),D)<r_0$ and so $\|\gamma(t)-D\|_\HS<r_1$ for all $0\le t\le1$, we
get
\begin{align*}
L_\phi(\gamma)&=\int_0^1\sqrt{\<\gamma'(t),\phi(\bL_{\gamma(t)},\bR_{\gamma(t)})^{-1}
\gamma'(t)\>_\HS}\,dt \\
&\ge\int_0^1\sqrt{\<\gamma'(t),(\phi(\bL_D,\bR_D)^{-1}-\rho\bI)
\gamma'(t)\>_\HS}\,dt \\
&=\int_0^1\|(\phi(\bL_D,\bR_D)^{-1}-\rho\bI)^{1/2}\gamma'(t)\|_\HS\,dt \\
&\ge\|(\phi(\bL_D,\bR_D)^{-1}-\rho\bI)^{1/2}(\eps H)\|_\HS \\
&=\eps\Big\|\bigl[(\phi(\lambda_i,\lambda_j)^{-1}-\rho)^{1/2}
\bigr]_{ij}\circ H\Big\|_\HS.
\end{align*}
In the above, note that $\phi(\bL_D,\bR_D)^{-1}-\rho\bI\ge0$ on the Hilbert space
$(\bM_n,\<\cdot,\cdot\>_\HS)$ since
$$
\rho<\min_i\lambda_i^{-\theta}=\min_{i,j}\phi(\lambda_i,\lambda_j)^{-1}.
$$
Also, the second inequality above follows since
$\int_0^1\|(\phi(\bL_D,\bR_D)^{-1}-\rho I)^{1/2}\gamma'(t)\|_\HS\,dt$ is the length in the
Euclidean space $(\bH_n,\|\cdot\|_\HS)$ and it is shortest if $\gamma$ is the segment
between $D$ and $D+\eps H$. Taking the infimum of $L_\phi(\gamma)$ gives
$$
\delta_\phi(D,D+\eps H)\ge\eps\Big\|\bigl[(\phi(\lambda_i,\lambda_j)^{-1}-\rho)^{1/2}
\bigr]_{ij}\circ H\Big\|_\HS.
$$
On the other hand, let $\gamma_0(t):=D+t\eps H$. Since
$\|\gamma_0(t)-D\|_\HS\le\eps\|H\|_\HS<r_1$ for $0\le t\le1$, we get
\begin{align*}
\delta_\phi(D,D+\eps H)&\le L_\phi(\gamma_0) \\
&=\int_0^1\sqrt{\<\gamma_0'(t),\phi(\bL_{\gamma_0(t)},\bR_{\gamma_0(t)})^{-1}
\gamma_0'(t)\>_\HS}\,dt \\
&\le\int_0^1\sqrt{\<\gamma_0'(t),(\phi(\bL_D,\bR_D)^{-1}+\rho\bI)
\gamma_0'(t)\>_\HS}\,dt \\
&=\|(\phi(\bL_D,\bR_D)^{-1}+\rho\bI)^{1/2}(\eps H)\|_\HS \\
&=\eps\Big\|\bigl[(\phi(\lambda_i,\lambda_j)^{-1}+\rho)^{1/2}
\bigr]_{ij}\circ H\Big\|_\HS.
\end{align*}
Since $\rho$ is arbitrary,
$$
\lim_{\eps\searrow0}{\delta_\phi(D,D+\eps H)\over\eps}
=\Big\|\bigl[\phi(\lambda_i,\lambda_j)^{-1/2}\bigr]_{ij}\circ H\Big\|_\HS
=\|\phi(\bL_D,\bL_D)^{-1/2}\circ H\|_\HS.
$$
\end{proof}

\noindent
{\it Proof of Theorem \ref{T-4.1}.}\enspace
First, (i) $\Leftrightarrow$ (ii) is easy to check. To prove (i) $\Rightarrow$ (iii), it
suffices to show that (i) implies that
\begin{equation}\label{F-4.1}
\|\phi^{(1)}(\bL_D,\bR_D)^{-1/2}H\|_\HS\ge\|\phi^{(2)}(\bL_D,\bR_D)^{-1/2}H\|_\HS
\end{equation}
for all $D\in\bP_n$ and $H\in\bH_n$. But this implication is immediately seen thanks to
\eqref{F-1.2}. (iii) $\Rightarrow$ (iv) is obvious. Finally, assume (iv) and apply
Lemma \ref{L-4.2} to get \eqref{F-4.1} for all $D\in\bP_n$ and $H\in\bH_n$. When
$D:=\bmatrix x&0\\0&y\endbmatrix\oplus I_{n-2}$ with $x,y>0$ and
$H:=\bmatrix1&1\\1&1\endbmatrix\oplus O_{n-2}$, \eqref{F-4.1} means that
$$
\sqrt{x^{-1}+2\phi^{(1)}(x,y)^{-1}+y^{-1}}
\ge\sqrt{x^{-1}+2\phi^{(2)}(x,y)^{-1}+y^{-1}},
$$
which gives (i).\qed

\begin{remark}\label{E-4.3}{\rm
Let $\cD_n:=\{D\in\bP_n:\Tr D=1\}$, a submanifold of $\bP_n$. One can replace
$(\bP_n,\bH_n)$ by $(\cD_n,\bH_n\ominus\bR I)$ and slightly modify the above proof to
show that the above (i)--(iv) are also equivalent to the following conditions reduced on
$\cD_n$:
\begin{itemize}
\item[\rm(iii$'$)] $L_{\phi^{(1)}}(\gamma)\ge L_{\phi^{(2)}}(\gamma)$ for all $C^1$ curve
in $\cD_n$;
\item[\rm(iv$'$)] $\delta_{\phi^{(1)}}^\cD(A,B)\ge\delta_{\phi^{(2)}}^\cD(A,B)$ for all
$A,B\in\cD_n$, where $\delta_\phi^\cD(A,B)$ denotes the geodesic distance in the
Riemannian manifold $(\cD_n,K^\phi)$.
\end{itemize}
}\end{remark}

By Theorems \ref{T-2.1} and \ref{T-4.1} we have:

\begin{cor}\label{C-4.4}
Let $M\in\MM_0$, $\theta\in\bR$ and $\phi(x,y):=M(x,y)^\theta$. If
$\phi(x,y)\le\phi_\theta(x,y)$ for all $x,y>0$ (see Section 2 for $\phi_\theta$), then for
every $A,B\in\bP_n$,
\begin{equation}\label{F-4.2}
\delta_\phi(A,B)\ge\delta_{\phi_\theta}(A,B)
=\begin{cases}
{2\over|2-\theta|}\|A^{2-\theta\over2}-B^{2-\theta\over2}\|_\HS
& \text{if $\theta\ne2$}, \\
\|\log A-\log B\|_\HS & \text{if $\theta=2$}.
\end{cases}
\end{equation}
If $\phi(x,y)\ge\phi_\theta(x,y)$ for all $x,y>0$, then the reversed inequality holds in
\eqref{F-4.2}.
\end{cor}

The next theorem is a refinement of Corollary \ref{C-4.4} with strict inequality under
additional assumptions.

\begin{thm}\label{T-4.5}
Let $M$, $\theta$ and $\phi$ be as in Corollary \ref{C-4.4}. Assume that $A,B\in\bP_n$ are
not commuting, i.e., $AB\ne BA$. If $\phi(x,y)<\phi_\theta(x,y)$ for all $x,y>0$ with
$x\ne y$, then $\delta_\phi(A,B)>\delta_{\phi_\theta}(A,B)$. Similarly,
$\delta_\phi(A,B)<\delta_{\phi_\theta}(A,B)$ if $\phi(x,y)>\phi_\theta(x,y)$ for all
$x,y>0$ with $x\ne y$.
\end{thm}

To prove the theorem, we need a simple lemma.

\begin{lemma}\label{L-4.6}
Let $\phi^{(k)}$, $k=1,2$, be as in Theorem \ref{T-4.1}, and assume that
$\phi^{(1)}(x,y)<\phi^{(2)}(x,y)$ for all $x,y>0$ with $x\ne y$. If $\gamma:[0,1]\to\bP_n$
is a $C^1$ curve and $\gamma(t)\gamma'(t)\ne\gamma'(t)\gamma(t)$ for some $t\in[0,1]$,
then $L_{\phi^{(1)}}(\gamma)>L_{\phi^{(2)}}(\gamma)$.
\end{lemma}

\begin{proof}
It suffices to show that if $D\in\bP_n$ and $H\in\bH_n$ are not commuting, then
$$
\|\phi^{(1)}(\bL_D,\bR_D)^{-1/2}H\|_\HS>\|\phi^{(2)}(\bL_D,\bR_D)^{-1/2}H\|_\HS.
$$
To prove this, we may assume that $D=\diag(\lambda_1,\dots,\lambda_n)$. Then $DH\ne HD$
means that $H_{ij}\ne0$ for some $(i,j)$ with $\lambda_i\ne\lambda_j$, where $H=[H_{ij}]$.
Since $\phi^{(1)}(\lambda_i,\lambda_j)<\phi^{(2)}(\lambda_i,\lambda_j)$ for such $(i,j)$, \
we obviously get
$$
\|\phi^{(1)}(\bL_D,\bR_D)^{-1/2}H\|_\HS^2
=\sum_{i,j=1}^n{|H_{ij}|^2\over\phi^{(1)}(\lambda_i,\lambda_j)}
>\sum_{i,j=1}^n{|H_{ij}|^2\over\phi^{(2)}(\lambda_i,\lambda_j)}
=\|\phi^{(2)}(\bL_D,\bR_D)^{-1/2}H\|_\HS^2,
$$
as required.
\end{proof}

\noindent
{\it Proof of Theorem \ref{T-4.5}.}\enspace
Assume that $\phi(x,y)<\phi_\theta(x,y)$ for all $x\ne y$ and on the contrary that
$\delta_\phi(A,B)=\delta_{\phi_\theta}(A,B)$. Choose a sequence $\{\gamma_k\}$ of $C^1$
curves from $A$ to $B$ such that $L_\phi(\gamma_k)\to\delta_{\phi_\theta}(A,B)$ as
$k\to\infty$. The following proof is given in the case $\theta\ne2$ but the case
$\theta=2$ is similar with obvious modifications. Let
$\xi_k(t):=\gamma_k(t)^{2-\theta\over2}$ for $0\le t\le1$. Since Theorem \ref{T-4.1} gives
$$
\delta_{\phi_\theta}(A,B)\le L_{\phi_\theta}(\gamma_k)\le L_\phi(\gamma_k)
\longrightarrow\delta_{\phi_\theta}(A,B)
$$
so that by Theorem \ref{T-2.1}
$$
{|2-\theta|\over2}L_{\phi_\theta}(\gamma_k)=\int_0^1\|\xi_k'(t)\|_\HS\,dt
\longrightarrow\|A^{2-\theta\over2}-B^{2-\theta\over2}\|_\HS
\quad\mbox{as $k\to\infty$}.
$$
By reparametrizing $\xi_k(t)$'s (hence $\gamma_k(t)$'s) one may assume that each $\xi_k$
has a constant speed, i.e.,
$$
\|\xi_k'(t)\|_\HS={|2-\theta|\over2}L_{\phi_\theta}(\gamma_k),
\qquad0\le t\le1.
$$
Set $\alpha:=\|A^{2-\theta\over2}-B^{2-\theta\over2}\|_\HS$ and
$H_0:=\alpha^{-1}\bigl(B^{2-\theta\over2}-A^{2-\theta\over2}\bigr)$, a unit vector in
$(\bH_n,\<\cdot,\cdot\>_\HS)$. We notice
\begin{align*}
\int_0^1\Biggl(1-\biggl\<{\xi_k'(t)\over\|\xi_k'(t)\|_\HS},H_0\biggr\>_\HS\Biggr)\,dt
&=1-{2\over|2-\theta|L_{\phi_\theta}(\gamma_k)}
\bigl\<B^{2-\theta\over2}-A^{2-\theta\over2},H_0\bigr\>_\HS \\
&=1-{2\alpha\over|2-\theta|L_{\phi_\theta}(\gamma_k)}
\longrightarrow0\quad\mbox{as $k\to\infty$}.
\end{align*}
Hence, by taking a subsequence, one can assume that
$$
\bigg\|{\xi_k'(t)\over\|\xi_k'(t)\|_\HS}-H_0\bigg\|_\HS^2
=2\Biggl(1-\biggl\<{\xi_k'(t)\over\|\xi_k'(t)\|_\HS},H_0\biggr\>_\HS\Biggr)
\longrightarrow0\quad\mbox{a.e. $t\in[0,1]$}.
$$
Since $\|\xi_k'(t)\|_\HS=2^{-1}|2-\theta|L_{\phi_\theta}(\gamma_k)\to\alpha$, this means
that
\begin{equation}\label{F-4.3}
\|\xi_k'(t)-\bigl(B^{2-\theta\over2}-A^{2-\theta\over2}\bigr)\|_\HS
\longrightarrow0\quad\mbox{a.e.\ $t\in[0,1]$},
\end{equation}
which implies also that for every $0\le t\le1$
\begin{align}\label{F-4.4}
\|\xi_k(t)-\bigl((1-t)A^{2-\theta\over2}+tB^{2-\theta\over2}\bigr)\|_\HS
&=\bigg\|\int_0^t\bigl(\xi_k'(s)-\bigl(B^{2-\theta\over2}-A^{2-\theta\over2}
\bigr)\bigr)\,ds\bigg\|_\HS \nonumber\\
&\le\int_0^t\|\xi_k'(s)-\bigl(B^{2-\theta\over2}-A^{2-\theta\over2}
\bigr)\|_\HS\,ds\longrightarrow0.
\end{align}

Now define $\xi_0(t):=(1-t)A^{2-\theta\over2}+tB^{2-\theta\over2}$ and
$\gamma_0(t):=\xi_0(t)^{2\over2-\theta}$. With $G_\theta(x):=x^{2\over2-\theta}$ one can
apply \eqref{F-2.7}, \eqref{F-4.3} and \eqref{F-4.4} to obtain
\begin{align*}
&\|\phi(\bL_{\gamma_0(t)},\bR_{\gamma_0(t)})^{-1/2}\gamma_0'(t)\|_\HS \\
&\qquad=\|\phi(\bL_{G_\theta(\xi_0(t))},\bR_{G_\theta(\xi_0(t))})^{-1/2}
G_\theta^{[1]}(\bL_{\xi_0(t)},\bR_{\xi_0(t)})\xi_0'(t)\|_\HS \\
&\qquad=\lim_{k\to\infty}
\|\phi(\bL_{G_\theta(\xi_k(t))},\bR_{G_\theta(\xi_k(t))})^{-1/2}
G_\theta^{[1]}(\bL_{\xi_k(t)},\bR_{\xi_k(t)})\xi_k'(t)\|_\HS \\
&\qquad=\lim_{k\to\infty}
\|\phi(\bL_{\gamma_k(t)},\bR_{\gamma_k(t)})^{-1/2}\gamma_k'(t)\|_\HS
\end{align*}
for a.e.\ $t\in[0,1]$. Fatou's lemma gives
\begin{equation}\label{F-4.5}
L_\phi(\gamma_0)\le\liminf_{k\to\infty}L_\phi(\gamma_k)
=\delta_{\phi_\theta}(A,B)=L_{\phi_\theta}(\gamma_0)
\end{equation}
thanks to Theorem \ref{T-2.1}. Here it is clear that $\xi_0(t)$ and $\xi_0'(t)$ are not
commuting for any $0\le t\le1$. Hence $\gamma_0(t)$ and $\gamma_0'(t)$ never commute for
$0\le t\le1$. In fact, this is seen because $\xi_0'(t)$ can be approximated by polynomials
of $\gamma_0(t)$ and $\gamma_0'(t)$ thanks to \eqref{F-2.6} applied to
$\xi_0(t)=G_\theta^{-1}(\gamma_0(t))$ so that
$\gamma_0(t)\gamma_0'(t)=\gamma_0'(t)\gamma_0(t)$ implies
$\xi_0(t)\xi_0'(t)=\xi_0'(t)\xi_0(t)$. Hence \eqref{F-4.5} contradicts the conclusion of
Lemma \ref{L-4.6}.

The proof of the second assertion is easy. Assume that $\phi(x,y)>\phi_\theta(x,y)$ for all
$x\ne y$, and let $\gamma_0(t)$ be same as in the proof of the first assertion. Since
$\gamma_0(t)$ and $\gamma_0'(t)$ never commute for $0\le t\le1$ as mentioned above, Lemma
\ref{L-4.6} again implies that
$$
\delta_\phi(A,B)\le L_\phi(\gamma_0)<L_{\phi_\theta}(\gamma_0)
=\delta_{\phi_\theta}(A,B),
$$
as required.\qed

\bigskip
The above proof of the first assertion is a bit involved. The proof would be much simpler
if a geodesic shortest path joining $A$ and $B$ exists in $(\bP_n,K^\phi)$, which is not
known at the moment.

\begin{example}\label{E-4.7}{\rm
The following are examples of the inequality given in Corollary \ref{C-4.4} in the cases
of familiar means. In fact, these are immediate consequences of Corollary \ref{C-4.4} and
Lemma \ref{L-2.2} together with \eqref{F-2.9}--\eqref{F-2.12} and Lemma \ref{L-2.5}.
Furthermore, Theorem \ref{T-4.5} shows that all inequalities in the following become strict
if $A,B$ are not commuting and the respective closed range of $\theta$ is replaced by the
open range.
\begin{itemize}
\item[(1)] For the $\theta$-power $M_\mathrm{A}^\theta(x,y)=\bigl({x+y\over2}\bigr)^\theta$
of the arithmetic mean,
$$
\delta_{M_\mathrm{A}^\theta}(A,B)
\begin{cases}\le\delta_{\phi_\theta}(A,B) & \text{if $\theta\le-2$, $\theta\ge0$}, \\
\ge\delta_{\phi_\theta}(A,B) & \text{if $-2\le\theta\le0$}.
\end{cases}
$$
\item[(2)] For the $\theta$-power
$M_{\sqrt{\phantom{a}}}^\theta(x,y)=\Bigl({\sqrt x+\sqrt y\over2}\Bigr)^{2\theta}$ of the
root mean,
$$
\delta_{M_{\sqrt{\phantom{a}}}^\theta}(A,B)
\begin{cases}\le\delta_{\phi_\theta}(A,B) & \text{if $\theta\le0$, $\theta\ge1$}, \\
\ge\delta_{\phi_\theta}(A,B) & \text{if $0\le\theta\le1$}.
\end{cases}
$$
\item[(3)] For the $\theta$-power
$M_\mathrm{L}^\theta(x,y)=\Bigl({x-y\over\log x-\log y}\Bigr)^\theta$ of the logarithmic
mean,
$$
\delta_{M_\mathrm{L}^\theta}(A,B)
\begin{cases}\le\delta_{\phi_\theta}(A,B) & \text{if $\theta\le0$, $\theta\ge2$}, \\
\ge\delta_{\phi_\theta}(A,B) & \text{if $0\le\theta\le2$}.
\end{cases}
$$
\item[(4)] For the $\theta$-power $M_\mathrm{G}^\theta(x,y)=(xy)^{\theta/2}$ of the
geometric mean,
$$
\delta_{M_\mathrm{G}^\theta}(A,B)
\begin{cases}\le\delta_{\phi_\theta}(A,B) & \text{if $\theta\le0$, $\theta\ge4$}, \\
\ge\delta_{\phi_\theta}(A,B) & \text{if $0\le\theta\le4$}.
\end{cases}
$$
\item[(5)] For the $\theta$-power
$M_\mathrm{H}^\theta(x,y)=\Bigl({2xy\over x+y}\Bigr)^\theta$ of the harmonic mean,
$$
\delta_{M_\mathrm{H}^\theta}(A,B)
\begin{cases}\le\delta_{\phi_\theta}(A,B) & \text{if $\theta\le0$}, \\
\ge\delta_{\phi_\theta}(A,B) & \text{if $0\le\theta\le10$}.
\end{cases}
$$
For any $\theta\in\bR$, $M_\mathrm{H}(x,1)<M_\theta(x,1)$ holds for large $x>0$ since
$\lim_{x\to\infty}M_\theta(x,1)=+\infty$ while $\lim_{x\to\infty}M_\mathrm{H}(x,1)=2$.
From this and Lemma \ref{L-2.5} we observe that $\delta_{M_\mathrm{H}^\theta}(A,B)$
and $\delta_{\phi_\theta}(A,B)$ are not comparable when $\theta>10$.
\end{itemize}
}\end{example}

In the case $\theta=2$ the above example (4) with \eqref{F-0.5} says that
$$
\|\log(A^{-1/2}BA^{-1/2})\|_\HS\ge\|\log A-\log B\|_\HS,
\qquad A,B\in\bP_n.
$$
This is the so-called {\it exponential metric increasing} ({\it EMI}\,) property in
\cite{Bh2,BH}.
On the other hand, for instance, (1) says that
$$
\delta_{M_\mathrm{A}^2}(A,B)\le\|\log A-\log B\|_\HS,
\qquad A,B\in\bP_n,
$$
which may be called the ``exponential metric decreasing" property. In the case $\theta=1$
the above examples give
$$
\delta_{M_\mathrm{H}}(A,B)\ge\delta_{M_\mathrm{G}}(A,B)
\ge\delta_{M_\mathrm{L}}(A,B)\ge2\|A^{1/2}-B^{1/2}\|_\HS
\ge\delta_{M_\mathrm{A}}(A,B),
$$
which may be called the ``square metric increasing/decreasing" properties.

In the particular case where $\phi(x,y)=M(x,y)$ (of degree $\theta=1$) is an operator
monotone mean, i.e., $M(x,1)$ is a standard operator monotone function and moreover $A,B$
are commuting, the next theorem gives the exact formula for $\delta_M(A,B)$ independently
of the choice of $M$. It seems that this independence of $M$ is reflected by the
uniqueness of a monotone Riemannian metric in the classical case (see \cite{Pe3}).

\begin{thm}\label{T-4.8}
Let $M\in\MM_0$ and assume that $M(x,1)$ is an operator monotone function. If $A,B\in\bP_n$
are commuting, then
$$
\delta_M(A,B)=2\|A^{1/2}-B^{1/2}\|_\HS,
$$
and a geodesic shortest curve from $A$ to $B$ is given by
$$
\gamma_{A,B}(t):=\bigl((1-t)A^{1/2}+tB^{1/2}\bigr)^2,\qquad0\le t\le1,
$$
independently of the choice of $M$ as above. Furthermore, this $\gamma_{A,B}$ is a unique
geodesic shortest curve from $A$ to $B$ whenever $M\ne M_\mathrm{A}$.
\end{thm}

First we give a small lemma.

\begin{lemma}\label{L-4.9}
Assume that $\gamma:[0,1]\to\bP_n$ is a $C^1$ curve and
$\gamma(t)\gamma'(t)=\gamma'(t)\gamma(t)$ for all $t\in[0,1]$. Let
$\xi(t):=\gamma(t)^{1/2}$. Then $L_M(\gamma)=2\int_0^1\|\xi'(t)\|_\HS\,dt$ for all $M$
as stated in Theorem \ref{T-4.8} (i.e., $M\in\MM_0$ with operator monotone $M(x,1)$).
\end{lemma}

\begin{proof}
Since $M(x,x)=x$ for all $x>0$, we note that $\|M(\bL_D,\bR_D)^{-1/2}H\|_\HS$ is
independent of the choice of $M$ whenever $D\in\bP_n$ and $H\in\bH_n$ are commuting. This
implies that $L_M(\gamma)$ is independent of $M$ if $\gamma$ is as stated in the lemma.
Hence the lemma follows by the $\theta=1$ case of Theorem \ref{T-2.1}.
\end{proof}

\noindent
{\it Proof of Theorem \ref{T-4.8}.}\enspace
Assume that $AB=BA$, and let $\gamma_{A,B}$ be as given in the theorem. By Lemma
\ref{L-4.9} and Theorem \ref{T-2.1} we have
$$
L_M(\gamma_{A,B})=L_{M_{\sqrt{\phantom{a}}}}(\gamma_{A,B})=2\|A^{1/2}-B^{1/2}\|_\HS
$$
so that $\delta_M(A,B)\le2\|A^{1/2}-B^{1/2}\|_\HS$. To prove the converse, let $\Phi$
denote the conditional expectation (with respect to $\Tr$) of $\bM_n$ onto
the commutative subalgebra generated by $A,B$, and let $\gamma:[0,1]\to\bP_n$ be an
arbitrary $C^1$ curve from $A$ to $B$. Then $\Phi(\gamma)$ is a $C^1$ curve in $\bP_n$
from $A$ to $B$. Since $K^M$ is a monotone metric \cite{Pe3} (see also Introduction),
we get
$$
K_{\Phi(\gamma(t))}^M(\Phi(\gamma'(t)),\Phi(\gamma'(t)))
\le K_{\gamma(t)}^M(\gamma'(t),\gamma'(t)),\qquad0\le t\le1,
$$
so that $L_\phi(\Phi(\gamma))\le L_\phi(\gamma)$. Hence we may assume that $\gamma(t)$'s
are in a commutative subalgebra. When $\xi(t):=\gamma(t)^{1/2}$, we get by
Lemma \ref{L-4.9}
$$
L_M(\gamma)=2\int_0^1\|\xi'(t)\|_\HS\,dt\ge2\|A^{1/2}-B^{1/2}\|_\HS.
$$
Hence $\delta_\phi(A,B)=2\|A^{1/2}-B^{1/2}\|_\HS$ and $\gamma_{A,B}$ is a common geodesic
shortest curve from $A$ to $B$ for all metrics $K^M$ with operator monotone $M$.

Next we show the last assertion on the uniqueness of a geodesic curve.  To prove this,
let $\gamma_1:[0,1]\to\bP_n$ be a $C^1$ curve from $A$ to $B$ such that
$L_M(\gamma_1)=2\|A^{1/2}-B^{1/2}\|_\HS$. Since $M_\mathrm{A}$ is the largest standard
operator monotone function and $M\ne M_\mathrm{A}$, note that $M(x,1)<M_\mathrm{A}(x,1)$
for all $x>0$ with $x\ne1$. Since $L_M(\gamma_1)\ge L_{M_\mathrm{A}}(\gamma_1)$ by
Theorem \ref{T-4.1}, it follows from Lemma \ref{L-4.6} that
$\gamma_1(t)\gamma_1'(t)=\gamma_1'(t)\gamma_1(t)$ for all $t\in[0,1]$. Lemma \ref{L-4.9}
in turn implies that $\int_0^1\|\xi_1'(t)\|_\HS\,dt=\|A^{1/2}-B^{1/2}\|_\HS$, where
$\xi_1(t):=\gamma_1(t)^{1/2}$. Therefore we get $\xi_1(t)=(1-t)A^{1/2}+tB^{1/2}$,
$0\le t\le 1$, so that $\gamma_1=\gamma_{A,B}$.\qed

\bigskip
When $M=M_\mathrm{A}$ and $A,B$ are commuting, it is not known whether $\delta_{A,B}$
is a unique geodesic shortest path joining $A,B$. To prove this, we probably need to
examine the equality case in the monotonicity of $K_D^M(H,H)$ under conditional
expectation. Another problem for commuting $A,B$ is whether $\delta_{A,B}$ gives a
geodesic shortest path for any metric $K^M$ with $M\in\MM_0$ which is not necessarily
operator monotone.

We close the section with a remark on comparison of skew informations given in
\eqref{F-0.9}. Let $f$ and $g$ be two standard operator monotone functions that are
regular, i.e., $f(0),g(0)>0$. It is immediate to see that $I_D^f(K)\ge I_D^g(K)$ for all
$D\in\bP_n$ and $K\in\bH_n$ if and only if $f(0)/f(x)\ge g(0)/g(x)$ for all $x>0$. For
example, as for $f_p=f_{1-p}$, $0<p\le1/2$, given in \eqref{F-0.8}, $f_p(0)/f_p(x)$ is
increasing in $p\in(0,1/2]$ so that the Wigner-Yanase-Dyson skew information
$I_D^\mathrm{WYD}(p,K)$ is increasing in $p\in(0,1/2]$ for fixed $D$ and $K$ (see
\cite{ACH}).

\section{Unitarily invariant norms}
\setcounter{equation}{0}

Let $|||\cdot|||$ be a {\it unitarily invariant norm} on matrices, that is, $|||\cdot|||$
is a norm on $\bM_n$, $n\in\bN$, such that $|||UXV|||=|||X|||$ for all $X,U,V\in\bM_n$
with $U,V$ unitaries. The Hilbert-Schmidt norm $\|\cdot\|_\HS$ is a special example of
such norms. When a kernel function $\phi:(0,\infty)\times(0,\infty)\to(0,\infty)$ is given,
replacing $\|\cdot \|_\HS$ by $|||\cdot|||$ in \eqref{F-1.3} we define the
{\it length}
$$
L_{\phi,|||\cdot|||}(\gamma)
:=\int_0^1|||\phi(\bL_{\gamma(t)},\bR_{\gamma(t)})^{-1/2}\gamma'(t)|||\,dt
$$
of a $C^1$ curve $\gamma:[0,1]\to\bP_n$. The {\it distance}
$\delta_{\phi,|||\cdot|||}(A,B)$ between $A,B\in\bP_n$ is the infimum of
$L_{\phi,|||\cdot|||}(\gamma)$ over all $C^1$ curves $\gamma$ from $A$ to $B$. The manifold
$\bP_n$ with the distance $\delta_{\phi,|||\cdot|||}$ is no longer a Riemannian manifold
but a certain Finsler manifold. When $|||\cdot|||$ is the operator norm, such Finsler
manifolds have been studied by several authors (see \cite{CPR} for example).

In this section we show that many results in the previous sections hold true even when the
Hilbert-Schmidt norm $\|\cdot \|_\HS$ is replaced by a general unitarily invariant norm
$|||\cdot|||$. First, Theorem \ref{T-2.1} can be extended as follows. We omit the proof
that is essentially same as the second proof of Theorem \ref{T-2.1}.

\begin{prop}\label{P-5.1}
Let $|||\cdot|||$ be any unitarily invariant norm. Let $M$, $\theta$, $\phi$ and $F$ be as
in Theorem \ref{T-2.1}. Then the transformation $D\in\bP_n\mapsto F(D)\in\bH_n$ is
isometric from $(\bP_n,\delta_{\phi,|||\cdot|||})$ into $(\bH_n,|||\cdot|||)$ if and only
if $F$ is in the form \eqref{F-2.1} and $M=M_\theta$ (so $\phi=\phi_\theta$). Moreover,
for every $A,B\in\bP_n$,
$$
\delta_{\phi_\theta,|||\cdot|||}(A,B)=\begin{cases}
{2\over|2-\theta|}|||A^{2-\theta\over2}-B^{2-\theta\over2}|||
& \text{if $\theta\ne2$}, \\
|||\log A-\log B||| & \text{if $\theta=2$},
\end{cases}
$$
and this distance is attained by curve \eqref{F-2.3}.
\end{prop}

The next comparison theorem is a partial extension of Theorem \ref{T-4.1}. An essential
point of the proof is similar to that of \cite[Theorem 1.1]{HK}.

\begin{prop}\label{P-5.2}
Let $M^{(1)},M^{(2)}\in\MM_0$, $\theta\in\bR$ and $\phi^{(k)}(x,y):=M_k(x,y)^\theta$,
$k=1,2$. Then the following conditions are equivalent:
\begin{itemize}
\item[\rm(i)] $(M^{(1)}(e^t,1)/M^{(2)}(e^t,1))^{\theta/2}$ is a positive definite function
on $\bR$;
\item[\rm(ii)] $L_{\phi^{(1)},|||\cdot|||}(\gamma)\ge L_{\phi^{(2)},|||\cdot|||}(\gamma)$
for all $C^1$ curve $\gamma$ in $\bP_n$ and for any unitarily invariant norm
$|||\cdot|||$;
\item[\rm(iii)] $L_{\phi^{(1)},\|\cdot\|_\infty}(\gamma)\ge
L_{\phi^{(2)},\|\cdot\|_\infty}(\gamma)$ for all $C^1$ curve $\gamma$ in $\bP_n$ and for
the operator norm $\|\cdot\|_\infty$.
\end{itemize}
\end{prop}

\begin{proof}
(i) $\Rightarrow$ (ii).\enspace
It suffices to show that (i) implies that
\begin{equation}\label{F-5.1}
|||\phi^{(1)}(\bL_D,\bR_D)^{-1/2}H|||\ge|||\phi^{(2)}(\bL_D,\bR_D)^{-1/2}H|||
\end{equation}
for all $D\in\bP_n$ and $H\in\bH_n$. To do this, one may assume that
$D=\diag(\lambda_1,\dots,\lambda_n)$. By \eqref{F-1.2} notice that
$$
\phi^{(2)}(\bL_D,\bR_D)^{-1/2}H=\Biggl[\biggl(
{\phi^{(1)}(\lambda_i,\lambda_j)\over\phi^{(2)}(\lambda_i,\lambda_j)}
\biggr)^{1/2}\Biggr]_{ij}\circ(\phi^{(1)}(\bL_D,\bR_D)^{-1/2}H)
$$
and
\begin{equation}\label{F-5.2}
\biggl({\phi^{(1)}(\lambda_i,\lambda_j)\over
\phi^{(2)}(\lambda_i,\lambda_j)}\biggr)^{1/2}
=\biggl({M^{(1)}(\lambda_i/\lambda_j,1)\over
M^{(2)}(\lambda_i/\lambda_j,1)}\biggr)^{\theta/2}
=\biggl({M^{(1)}(e^{\log\lambda_i-\log\lambda_j},1)\over
M^{(2)}(e^{\log\lambda_i-\log\lambda_j},1)}\biggr)^{\theta/2}.
\end{equation}
Since (i) implies that
$\bigl[(\phi^{(1)}(\lambda_i,\lambda_j)/\phi^{(2)}(\lambda_i,\lambda_j)^{1/2}\bigr]_{ij}$
is a positive definite matrix with all diagonal entries equal to $1$, \eqref{F-5.1}
is obtained (see \cite[1.4.1]{Bh3} for example).

(ii) $\Rightarrow$ (iii) is trivial.

(iii) $\Rightarrow$ (i).\enspace
For $k=1,2$, since $D\in\bP_n\mapsto\phi^{(k)}(\bL_D,\bR_D)$ is continuous, it is obvious
that
$$
\lim_{\eps\searrow0}{L_{\phi^{(k)},\|\cdot\|_\infty}([D,D+\eps H])\over\eps}
=\|\phi^{(k)}(\bL_D,\bR_D)^{-1/2}H\|_\infty
$$
for all $D\in\bP_n$ and $H\in\bH_n$, where $[D,D+\eps H]$ denotes the straight segment
$D+t\eps H$, $0\le t\le1$. Hence condition (iii) implies that
$$
\|\phi^{(1)}(\bL_D,\bR_D)^{-1/2}H\|_\infty\ge\|\phi^{(2)}(\bL_D,\bR_D)^{-1/2}H\|_\infty,
\qquad D\in\bP_n,\ H\in\bH_n.
$$
When $D=\diag(\lambda_1,\dots,\lambda_n)$, this means that
$$
\|H\|_\infty\ge\Bigg\|\Biggl[\biggl(
{\phi^{(1)}(\lambda_i,\lambda_j)\over\phi^{(2)}(\lambda_i,\lambda_j)}
\biggr)^{1/2}\Biggr]_{ij}\circ H\Bigg\|_\infty,
\qquad H\in\bH_n.
$$
Now the proof of \cite[Theorem 1.1]{HK} shows that
$\bigl[(\phi^{(1)}(\lambda_i,\lambda_j)/\phi^{(2)}(\lambda_i,\lambda_j))^{1/2}\bigr]_{ij}$
is positive semidefinite, which means (i) thanks to \eqref{F-5.2}.
\end{proof}

\begin{remark}\label{R-5.3}{\rm
The geodesic distance versions of the above (ii) and (iii) are
\begin{itemize}
\item[(iv)] $\delta_{\phi^{(1)},|||\cdot|||}(A,B)\ge\delta_{\phi^{(2)},|||\cdot|||}(A,B)$ for
all $A,B\in\bP_n$ and for any unitarily invariant norm $|||\cdot|||$;
\item[(v)] $\delta_{\phi^{(1)},\|\cdot\|_\infty}(A,B)\ge
\delta_{\phi^{(2)},\|\cdot\|_\infty}(A,B)$ for all $A,B\in\bP_n$.
\end{itemize}
Obviously, (ii) $\Rightarrow$ (iv) and (iii) $\Rightarrow$ (v). It may be expected that
(iv) and (v) are also equivalent to the conditions of Proposition \ref{P-5.2}. This would
be proved as in the proof of (iii) $\Rightarrow$ (i) if we have
$$
\lim_{\eps\searrow0}{\delta_{\phi,\|\cdot\|_\infty}(D,D+\eps H)\over\eps}
=\|\phi(\bL_D,\bR_D)^{-1/2}H\|_\infty
$$
for all $D\in\bP_n$, $H\in\bH_n$ and for $\phi=M^\theta$ with $M\in\MM_0$. Although the
above convergence for $\|\cdot\|_\HS$ is Lemma \ref{L-4.2}, we do not know whether it is
also true for $\|\cdot\|_\infty$.
}\end{remark}

For $M^{(1)},M^{(2)}\in\MM_0$ consider the following conditions:
\begin{itemize}
\item[(a)] $M^{(1)}(x,1)\le M^{(2)}(x,1)$ for all $x>0$;
\item[(b)] $M^{(1)}(e^t,1)/M^{(2)}(e^t,1)$ is positive definite on $\bR$ (in this case we
write $M^{(1)}\preceq M^{(2)}$);
\item[(c)] $M^{(1)}(e^t,1)/M^{(2)}(e^t,1)$ is infinitely divisible in the sense that
$(M^{(1)}(e^t,1)/M^{(2)}(e^t,1))^r$ is positive definite on $\bR$ for any $r>0$ (in this
case we write $M^{(1)}\ll M^{(2)}$).
\end{itemize}

Obviously, (c) $\Rightarrow$ (b) $\Rightarrow$ (a). Condition (a) appeared in Theorem
\ref{T-4.1} while (b) is in the case $\theta=2$ of Proposition \ref{P-5.2}. We also note
that (b) played an essential role in \cite{HK,HK2}. It was recently observed in
\cite{BK,Ko1} that the stronger condition (c) is even satisfied for many cases where
$M^{(1)},M^{(2)}\in\MM_0$ satisfy (b). In fact, Kosaki \cite{Ko3} communicated to us that
$$
M_\mathrm{H}\ll M_\mathrm{G}\ll M_\mathrm{L}\ll M_{\sqrt{\phantom{a}}}\ll M_\mathrm{A}
$$
can be easily shown by applying \cite[Corollary 3]{Ko1} and \cite[Proposition 4]{BK}.
Hence by Proposition \ref{P-5.2} (also Remark \ref{R-5.3}), if $\theta\ge0$ then
$$
\delta_{M_\mathrm{H}^\theta,|||\cdot|||}(A,B)
\ge\delta_{M_\mathrm{G}^\theta,|||\cdot|||}(A,B)
\ge\delta_{M_\mathrm{L}^\theta,|||\cdot|||}(A,B)
\ge\delta_{M_{\sqrt{\phantom{a}}}^\theta,|||\cdot|||}(A,B)
\ge\delta_{M_\mathrm{A}^\theta,|||\cdot|||}(A,B),
$$
and inequalities are reversed if $\theta\le0$. For $\{N_\alpha\}_{0\le\alpha\le2}$ given
in \eqref{F-3.1}, if $0\le\alpha<\beta\le2$ then we have $N_\beta\ll N_\alpha$ by
\cite[Theorem 2]{BK} since
${N_\beta(e^{2t},1)/N_\alpha(e^{2t},1)}=(\beta/\alpha)(\sinh{\alpha t}/\sinh{\beta t})$.
As for $\psi_\alpha=N_\alpha^2$, similarly to Theorem \ref{T-3.3} we have
$$
\delta_{\psi_\alpha,|||\cdot|||}(A,B)
=|||\log(A^{-\alpha/2}B^\alpha A^{-\alpha/2})^{1/\alpha}|||,
\qquad0<\alpha\le2,
$$
which decreases to $\delta_{M_\mathrm{L}^2,|||\cdot|||}(A,B)=|||\log A-\log B|||$ as
$\alpha\searrow0$ (this is also a consequence of Araki's log-majorization \cite{Ar} as
mentioned at the end of Section 3). In particular, the inequality
$$
\delta_{M_\mathrm{G}^2,|||\cdot|||}(A,B)=|||\log(A^{-1/2}BA^{-1/2})|||
\ge|||\log A-\log B|||
$$
is the generalized EMI in \cite{Bh2}.

Finally, as for $\phi_\theta$ we show:

\begin{prop}\label{P-5.4}
Let $|||\cdot|||$ be any unitarily invariant norm and $A,B\in\bP_n$. Then
$\delta_{\phi_\theta,|||\cdot|||}(A,B)$ (see Proposition \ref{P-5.1}) is decreasing in
$\theta\in(-\infty,2]$ and increasing in $\theta\in[2,\infty)$. Furthermore,
$$
\delta_{M_\mathrm{G}^\theta,|||\cdot|||}(A,B)
\begin{cases}\le\delta_{\phi_\theta,|||\cdot|||}(A,B)
& \text{if $\theta\le0$, $\theta\ge4$}, \\
\ge\delta_{\phi_\theta,|||\cdot|||}(A,B) & \text{if $0\le\theta\le4$}.
\end{cases}
$$
\end{prop}

\begin{proof}
Assume that $\theta'<\theta<2$ or $2<\theta<\theta'$, and define a kernel function
$k:(0,\infty)\times(0,\infty)\to(0,\infty)$ by
$$
k(x,y):={2-\theta'\over2-\theta}\cdot
{x^{2-\theta\over2}-y^{2-\theta\over2}\over x^{2-\theta'\over2}-y^{2-\theta'\over2}}.
$$
The kernel $k(x,y)$ is positive definite (even infinitely divisible) by
\cite[Theorem 2]{BK} and $k(x,x)=1$ for all $x>0$. With the diagonalizations
$A=U\diag(\lambda_1,\dots,\lambda_n)U^*$ and $B=V\diag(\mu_1,\dots,\mu_n)V^*$ we write
\begin{align*}
{2\over2-\theta}\Bigl(A^{2-\theta\over2}-B^{2-\theta\over2}\Bigr)
&=U\biggl({2\over2-\theta}\biggl[\lambda_i^{2-\theta\over2}
-\mu_j^{2-\theta\over2}\biggr]_{ij}\circ(U^*V)\biggr)V^* \\
&=U\biggl(\bigl[k(\lambda_i,\mu_j)\bigr]_{ij}\circ
{2\over2-\theta'}\biggl[\lambda_i^{2-\theta'\over2}
-\mu_j^{2-\theta'\over2}\biggr]_{ij}\circ(U^*V)\biggr)V^* \\
&=U\biggl(\bigl[k(\lambda_i,\mu_j)\bigr]_{ij}\circ
U^*{2\over2-\theta'}\Bigl(A^{2-\theta'\over2}-B^{2-\theta'\over2}\Bigr)V\biggr)V^*.
\end{align*}
Hence \cite[1.4.1]{Bh3} can be applied to obtain
$\delta_{\phi_\theta,|||\cdot|||}(A,B)\le\delta_{\phi_{\theta'},|||\cdot|||}(A,B)$ thanks
to Proposition \ref{P-5.1}.

The second assertion (extending (4) of Example \ref{E-4.7}) follows since
$M_\theta\ll M_\mathrm{G}$ for $\theta\ge4$ and $M_\mathrm{G}\ll M_\theta$ for
$\theta\le4$ (see \cite[\S2.6]{BK}).
\end{proof}

\section*{Acknowledgments}

This work is partially supported by the Hungarian Research Grant OTKA T068258 (D.P.) and
Grant-in-Aid for Scientific Research (B)17340043 (F.H.) as well as by Hungary-Japan
HAS-JSPS Joint Project (D.P.\ \& F.H.).
D.P.\ thanks to Professors Peter Michor and Gabor Toth for communication about Riemannian
metrics. F.H.\ thanks to Professor Hideki Kosaki for communication about operator
monotonicity and infinite divisibility for means.

\end{document}